\documentclass{article} 
\usepackage{iclr2018_conference,times}
\usepackage{hyperref,url}
\usepackage{wrapfig}
\usepackage{float}
\usepackage{amsmath,amssymb}
\usepackage{float}
\usepackage{bbm}
\usepackage{color}
\usepackage{hyperref}
\usepackage[ruled]{algorithm2e}
\usepackage{amsthm,amssymb,amsmath,mathtools}
\usepackage{algcompatible}
\usepackage{algpseudocode}
\usepackage{graphicx}
\graphicspath{ {images/} }
\RequirePackage{filecontents}
\usepackage{subcaption} 
\usepackage{graphicx} 
\usepackage{subcaption}
\usepackage{dsfont}
\usepackage{setspace}
\usepackage{sectsty}
\usepackage{hyperref}
\usepackage{bbm}
\newcommand{\pth}[1]{\left( #1 \right) }

\newcommand*{\argmin}{\text{argmin}}
\newcommand*{\argmax}{\text{argmax}}

\newcommand*{\minl}{\min\limits}
\newcommand*{\maxl}{\max\limits}
\newcommand*{\infl}{\inf\limits}

\newcommand{\abs}[1]{\left|#1\right|}
\newcommand{\cut}[1]{\text{cut}\left(#1\right)}
\newcommand{\set}[2]{\left\{ #1 : #2 \right\}}

\newcommand{\E}{\operatorname{\mathbb{E}}}
\newcommand{\Graph}[2]{G=\left(#1,#2\right)}
\newcommand{\Loss}[2]{L_b^{#1}\left(#2\right)}

\newcommand{\ceil}[1]{\left\lceil #1 \right\rceil}
\newcommand{\curly}[1]{\left\{ #1 \right\}}
\newcommand{\braces}[1]{\left\{ #1 \right\}}
\newcommand{\GR}[2]{GR_{#2}\left(#1\right)}
\newcommand{\EGR}[2]{EGR_{#2}\left(#1\right)}
\newcommand{\MGR}[2]{MGR_{#2}\left(#1\right)}
\newcommand{\mgr}[2]{mgr_{#2}\left(#1\right)}
\newcommand{\ESR}[2]{ESR_{#2}\left(#1\right)}
\newcommand{\MSR}[2]{MSR_{#2}\left(#1\right)}
\newcommand{\msr}[2]{msr_{#2}\left(#1\right)}
\newcommand{\FOAP}{\pi^F}
\newcommand{\UC}[2]{\mathcal{U}^{#1}_{#2}}
\newcommand{\subplus}[1]{\left[ ~#1~ \right]_+}
\newcommand{\indicator}[1]{ \mathbbm{1}\curly{#1} }

\makeatletter
\newcommand{\@giventhatstar}[2]{\left(#1\,\middle|\,#2\right)}
\newcommand{\@giventhatnostar}[3][]{#1\left(#2\,#1|\,#3#1\right)}
\newcommand{\giventhat}{\@ifstar\@giventhatstar\@giventhatnostar}
\makeatother

\usepackage{pgf,tikz}
\newtheorem*{AuxLemma}{Aux. Lemma}
\newtheorem{remark}{Remark}

\newtheorem{theorem}{Theorem}[section]
\newtheorem{corollary}{Corollary}[section]
\newtheorem{lemma}{Lemma}[section]

\newtheorem{definition}{Definition}[section]

\usetikzlibrary{arrows, automata, matrix}
\usetikzlibrary{shapes.geometric}
\usepackage{pgfplots}

\usepackage[ruled]{algorithm2e}

\usepackage{multirow}

\title{The Stochastic Firefighter Problem}

\author{Guy Tennenholtz\\
	Technion - Israel Institute of Technology
\\
\AND{Constantine Caramanis} \\
The University of Texas at Austin \\
\AND{Shie Mannor} \\
Technion - Israel Institute of Technology
}

\iclrfinalcopy 

\begin{document}
	\maketitle
	\begin{abstract}
The dynamics of infectious diseases spread is crucial in determining their risk and offering ways to contain them. We study sequential vaccination of individuals in networks.  In the original (deterministic) version of the Firefighter problem, a fire breaks out at some node of a given graph. At each time step, b nodes can be protected by a firefighter and then the fire spreads to all unprotected neighbors of the nodes on fire. The process ends when the fire can no longer spread. We extend the Firefighter problem to a probabilistic setting, where the infection is stochastic. We devise a simple policy that only vaccinates neighbors of infected nodes and is optimal on regular trees and on general graphs for a sufficiently large budget.  We derive methods for calculating upper and lower bounds of the expected number of infected individuals, as well as provide estimates on the budget needed for containment in expectation. We calculate these explicitly on trees, d-dimensional grids, and Erd\H{o}s R\'{e}nyi graphs. Finally, we construct a state-dependent budget allocation strategy and demonstrate its superiority over constant budget allocation on real networks following a first order acquaintance vaccination policy. 
	\end{abstract}
	
\section{Introduction}
Consider an outbreak of a fatal disease. Flu outbreaks happen every year and vary in severity, depending in part on what type of virus is spreading. The influenza or flu pandemic of 1918 was the deadliest in modern history, infecting over 500 million people worldwide and killing 20-50 million victims. Between the years 2014-2016 West Africa experienced the largest outbreak of Ebola in history, with multiple countries affected. A total of 11,310 deaths were recorded in Guinea, Liberia, and Sierra Leone. With such pandemics, global connectedness may trigger a cascade of infections. The outbreak of Ebola which raged in West Africa echoes a scenario where long-range routes of transmission - most prominently, international air routes - can allow the deadliest viral strains to outrun their own extinction, and in the process kill vastly more victims than they would have otherwise. Successful pathogens leave their hosts alive long enough to spread infection. Supplies for vaccinating the vast population may be scarce as opposed to the contagion speed. We search for immunization strategies, attempting to find methods to eliminate epidemic threats by vaccinating parts of a network. Infections may propagate quickly and be discovered only at late stages of propagation.

The analysis of epidemic spreading in networks has produced results of practical importance, but only recently the study of epidemic models under dynamic control has begun. Vaccination policies define rules for identification of individuals that should be made immune to the spreading epidemic. In this paper we study  sequential vaccination policies that use full information on the infectious network's state and topology, under specific budget constraints. The problem of vaccination is central in the study of epidemics due to its practical implications. Viruses, sickness, and opinions can propagate through complex networks, influencing society, with benefits, but also calamitous effects. Viruses spreading over phones and email, online shaming over social networks, human disease via multiple relationships, are merely a few examples of the fundamental necessity of epidemic research.


\subsection{Our Contribution and Related Work}
\label{relationWorkSection}

The Firefighter problem was first introduced by \cite{hartnell1995firefighter}. It is a deterministic, discrete-time model of the spread of a fire on the nodes of a graph.  In its original version, a fire breaks out at some node of a given graph. At each time step, $b$ nodes can be protected by a firefighter and then the fire spreads to all unprotected neighbors of the nodes on fire (this protection is permanent). The process ends when the fire can no longer spread. At the end, all nodes that are not on fire are considered saved. The objective is at each time step to choose a node that will be protected by a firefighter such that a maximum number of nodes in the graph is saved at the end of the process. 
The Firefighter problem has received considerable attention (see, e.g., \cite{anshelevich2012approximability}, \cite{cai2008firefighting}, \cite{develin2004fire}, \cite{iwaikawa2011improved}, \cite{king2010firefighter}, \cite{ng2008generalization}). Develin and Hartke \cite{develin2007fire} proved a previous conjecture that $2d-1$ firefighters per time step are needed to contain a fire outbreak starting at a single node in the d-dimensional grid when $d \geq 3$. They also found an optimal solution to the case of the 2d grid, with 2 firefighters per time step. Their optimal solution is depicted in Figure \ref{fig:grid2d_optimal}. The Firefighter problem has been found to be infamously difficult. It is known to be NP-Complete, even when restricted to bipartite graphs (\cite{macgillivray2003firefighter}) or to trees of maximum degree three (\cite{finbow2007firefighter}), and was proven to be NP-hard to approximate within $n^{1-\epsilon}$ for any $\epsilon > 0$ (\cite{anshelevich2009approximation}). A survey of results on the Firefighter problem can be found in \cite{finbow2009firefighter}.

\begin{figure*}
\begin{center}
\centerline{
\includegraphics[width=0.6\columnwidth]{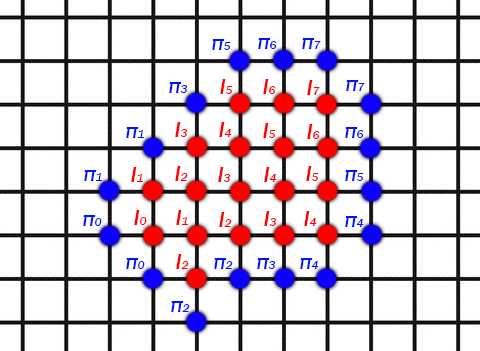}
}
\caption{Solution to the integer program used in \cite{develin2007fire} to prove the optimal solution to the (deterministic) Firefighter problem on the 2d grid for the case of a single source with a budget of $b = 2$. The fire outbreak starts at time $t=0$ at the root ($I_0$), and then spreads to the red nodes labeled $I_t$, where $t$ indicates infection time. The blue nodes $\pi_t$ are the nodes which are defended by firefighters at time $t$. This placement of two firefighters per time step completely contains the outbreak in 8 time steps, allowing a minimum number of 18 burnt nodes.}
\label{fig:grid2d_optimal}
\end{center}
\vskip -0.2in
\end{figure*}

The control of epidemics has been extensively studied for the past two decades. The dynamic allocation of cure has been studied in \cite{gourdin2011optimization,chung2009distributing,borgs2010distribute,drakopoulos2014efficient,drakopoulos2015network}. The Firefighter problem \cite{hartnell1995firefighter,develin2007fire,esperet2013fire,finbow2007firefighter,anshelevich2009approximation,macgillivray2003firefighter,king2010firefighter,cai2010surviving} and this paper do not model healing processes of individuals, making the task of finding an optimal vaccination strategy dependent only on the spreading infection process and the underlying graph. Not including a ``natural healing" process while only allowing vaccination of non-infected individuals, differentiates the Firefighter problem and our problem from these stochastic epidemic models.
More related to our work, are the studies of vaccine allocation in \cite{cohen2003efficient,preciado2013optimal,peng2013vaccination,ruan2012epidemic,miller2007effective,peng2010epidemic,bodine2011minimizing}. Victor M. Preciado et al. \cite{preciado2013optimal}, propose an optimization strategy for optimal vaccine allocation. Their model assumes slightly modifiable infection rates and a cost function based on a mean field approach. Information-driven vaccination is studied in \cite{ruan2012epidemic}, showing the spread of the information will promote people to take preventive measures and consequently suppress the epidemic spreading. Miller JC and Hyman JM propose in \cite{miller2007effective} to only vaccinate those nodes with the most unvaccinated contacts. An acquaintance immunization policy was proposed in \cite{cohen2003efficient}, where it was shown that such a policy is efficient in networks with broad-degree distribution. A different approach \cite{bodine2011minimizing}, considers minimizing the social cost of an epidemic. The above studies' analysis is based on epidemic thresholds and mean-field approximations of the evolution process. In contrast to these, this paper can be viewed as a stochastic extension to the Firefighter problem, which studies the transient, short-term behaviour of a spreading infection. 

Another variant of the Firefighter problem proposed by \cite{anshelevich2012approximability}, assumes the vaccination is also a process that spreads through the network. In the spreading vaccination model, vaccinated nodes propagate through the network as well. If a node $v$ is adjacent to a vaccinated node, then $v$ itself becomes vaccinated during the next time step (unless it is already infected or vaccinated). This type of model may depict conflicting ideas co-existing in a social network or the fact that vaccines can be infectious as well, since they are often an attenuated version of the actual disease. Anshelevich el al. \cite{anshelevich2012approximability} show that the ``spreading vaccination'' problem can be written as a maximization problem of a monotone submodular function over a matroid. These optimization problems are known to have a greedy algorithm which achieves a constant factor approximation \cite{krause2014submodular}. The reader may refer to \cite{anshelevich2012approximability} for an exhaustive account of the problem. This paper focuses on the original Firefighter problem, without spreading vaccination.

Previous research on immunization policies modeled the problem either deterministically or approximately, considering long-term effects alone. Defending a network from an attack of a virus must take into account its transient behavior, fast propagation, as well as the natural budget constraints of defending the network. It is thus vital to find on-line policies which acknowledge short-term effects and the defender's limited protection capabilities. This differentiates our paper from previous research, and takes a leap towards exact analysis, under a transient probabilistic framework.

\noindent \textbf{Our contributions are as follows.}
\begin{enumerate}
\item We extend the original Firefighter problem to a stochastic framework and define optimality criteria for (a) number of infected individuals and (b) budget needed for containment.
\item We propose several approaches for analyzing immunization and containment of ``fast-moving" (i.e., not infinitesimally slow) epidemics in networks. This is in contrast to the common assumption of small infection rates. We then develop upper and lower bounds for these.
\item We construct an algorithm for state dependent budget (i.e., budget allocation which changes according to the state of infected and vaccinated individuals at time $t$).
\end{enumerate}

\noindent \textbf{The paper is organized as follows.}
\begin{enumerate}
\item In Sections \ref{modelSection} and \ref{preliminariesSection} we define the model and optimality criteria, as well as define graph preliminaries that are used throughout our derivations.
\item In Section \ref{fopSection} we propose a simple greedy policy, which only chooses to vaccinate nodes neighboring the infection. An acquaintance immunization policy was proposed in \cite{cohen2003efficient}, where it was shown that such a policy is efficient in networks with broad-degree distribution. This greedy policy can be efficiently calculated on any graph, and is thus of high interest. We show such a policy is optimal when either the network topology is a regular tree, or the graph is finite and the infection spread is sufficiently slow.
\item In Section \ref{GrowthRateSection} we develop mathematical machinery to characterize the topology-dependent speed. We define growth rates, \mbox{maximal $/$ minimal} growth rates, as well as expected growth rates. We show in Theorem \ref{MGRThm} how these quantities can be used to bound the expected limiting cardinality of an infection under any vaccination policy.
\item Section \ref{containmentSection} considers criteria for containment of an infection. In Theorem \ref{containmentThm1} we construct upper and lower estimates on the budget needed to contain an infection.
\item We apply the results of Theorem \ref{containmentThm1} to specific network topologies, including regular trees, $d$-dimensional grids, and Erd\H{o}s-R\'{e}nyi graphs to obtain explicit bounds on the budget needed for containment. We summarize these results in Table \ref{table:summary}.
\item In Section \ref{stateDependentBudgetSection} we propose an algorithm for state-dependent budget allocation, based on estimated local growth rate behavior of the network. We show this strategy achieves better containment on two real world networks, when compared to a constant budget strategy which consumes an equal global budget.
\end{enumerate}

\section{Model}
\label{modelSection}
This section defines the model of our problem and the criteria we wish to optimize. A natural way to extend the Firefighter problem to a probabilistic setting is to model it as an MDP (\cite{bertsekas1995dynamic}). The NP-Completeness of the Firefighter problem impedes much for its research. Modeling the problem as probablistic, as well as defining different objectives to the problem, may allow us to obtain stronger results. We will distinguish between two different criteria. The first, and most natural criterion, considers minimizing the number of infected nodes at the end of the infection process. This criterion, which is the same as the original Firefighter problem objective, is hard to optimize generally. We consider an alternative criterion, which is ultimately easier to evaluate. The containment criterion asks the following question: What is the minimal budget needed to ensure an infection won't reach a certain size? This question can be answered immediately once an optimal policy is known, whereas knowing the latter does not give us the former. Strictly speaking, this question may be more feasible to answer, and is thus the main focus of our work. In regard to the optimality criterion, we find upper and lower bounds on the optimal loss.

\subsection{Stochastic Firefighter Model as an MDP}

We model the problem as a discrete SIR epidemic model, defined by parameters $G$ (graph network topology), $p \in \left(0,1\right]$ (infection probability / speed), and ${b \in \curly{0, 1,2,3,\hdots}}$ (vaccination budget per stage). We consider a network, represented by an undirected graph $\Graph{V}{E}$, where $V$ denotes the set of nodes and $E$ denotes the set of edges. Two nodes are said to be neighbors if $\left(u,v\right) \in E$. We use the notation $u \leftrightarrow v$ to denote neighboring nodes. We use $n$ to denote the number of nodes in $G$, and do not restrict ourselves to finite graphs.

Assume a spreading infection on $G$ to be a discrete time contact process, modeled by a discrete time Markov chain $\left \{s_t\right \}_{t=0}^{\infty}$. An infected node infects each of its healthy neighbors with probability $p$. A state of the chain, $s_t$, is defined as the pair $s_t = \left(I_t, B_t\right)$, where $I_t$ and $B_t$ denote the set of infected and vaccinated nodes at time $t$, respectively. At all times, $I_t \cap B_t = \emptyset$ (i.e., a node can either be healthy, infected, or vaccinated). 

The process $s_t$ is initialized at some given state $s_0 = (I_0, B_0)$ with transitions occurring independently according to the following dynamics
\begin{enumerate}
\item If a node $i$ is infected, it remains infected forever.
\item If a node $i$ is vaccinated, it remains vaccinated forever.
\item If a node $i$ is healthy at time $t$, then node $i$ stays healthy at time $t+1$ with probability $q_{i,t}$, where
\begin{equation*}
q_{i,t} = \left(1-p\right) ^ { \abs{\set{j}{i \leftrightarrow j, j \in I_t} } }.
\end{equation*}
This in turn means node $i$ moves to the infected state  at time $t+1$ with probability
\begin{equation*}
1 - q_{i,t} = 1 - \left(1-p\right) ^ { \abs{\set{j}{i \leftrightarrow j, j \in I_t} } }.
\end{equation*}
\end{enumerate}
Informally, conditions 1 and 2 mean that nodes remain in their infected or vaccinated states at all times. Condition 3 means that at each iteration nodes infect each of their healthy neighbors independently with probability $p$.

Our MDP is defined by a state space $\mathcal{S}$ given by tuples $(I,B)$ that are subsets of $V\times V$, action space $\mathcal{A}$ which is a subset of $V$ of maximum cardinality $b$, a transition probability matrix as defined by the infection dynamics, and a cost function with a penalty of $1$ for each infected node. A stationary Markov control policy $\pi: \mathcal{S} \to \mathcal{A}$ determines at each state $s = (I,B)$ what set of healthy nodes to vaccinate. Once a set of nodes is chosen to be vaccinated they are added to $B$. We impose a budget constraint of the form
\begin{equation*}
\abs{\pi(s)} \leq b
\end{equation*}
to all $s \in \mathcal{S}$. More specifically, our action space $\mathcal{A}$ is a subset of the healthy nodes in state $s$ with cardinality at most $b$.

To remove ambiguity in regard to the order of transitions, we assume that at each time step $t$, policy $\Pi_{Sb}$ first chooses the set of nodes to vaccinate $\pi(s_t)$, after which remaining healthy nodes are infected by their neighbors according to the contagion dynamics.
We denote the set of all stationary policies by $\Pi_S$.

Since any reasonable policy uses all of its budget at each iteration, we restrict ourselves to a subset of $\Pi_S$, which we denote by $\Pi_{Sb}$, such that all $\pi \in \Pi_{Sb}$ vaccinate exactly $b$ nodes at each iteration. 
That is, 
\begin{equation*}
\abs{\pi(s)} = b.
\end{equation*}

 \subsection{Optimality Objective} 
 \label{optimalityCriteriaSection}

We define an optimal policy as one which minimizes the expected number of infected nodes once there are no more nodes to vaccinate and infect. Formally, for an initial state ${s_0 = \left(I_0, B_0\right)}$, we define the expected loss of a policy $\Pi_{Sb}$ by
\begin{equation}
\label{LossDef}
\Loss{\pi}{s_0} = \limsup\limits_{T \to \infty} \E^{\pi} \giventhat{\frac{1}{T} \sum_{t=0}^{T} \abs{I_t}}{s_0},
\end{equation}
where $\E ^\pi$ denotes the expected value induced by policy $\Pi_{Sb}$.
For finite graphs, as $t \to \infty$, an absorbing state will be reached  almost surely. Then, the loss function can also be written as
\begin{equation*}
\Loss{\pi}{s_0} = \E^{\pi} \giventhat{\abs{I_{T*}}}{s_0},
\end{equation*}
where $T^*$ is a random variable that depicts the time in which an absorbing state is reached under policy $\Pi_{Sb}$ (i.e., when there are no more healthy neighbors to infect).
When $B_0 = \emptyset$ we  write $\Loss{\pi}{I_0}$ in place of $\Loss{\pi}{s_0}$.
We define the optimal loss $L^*$ by
\begin{equation*}
\Loss{*}{s_0} = \inf_{\pi \in \Pi_{Sb}} \Loss{\pi}{s_0},
\end{equation*}
When it exists, we define an optimal policy $\pi^*$ by
\begin{equation*}
\pi^* \in \argmin_{\pi \in \Pi_{Sb}} \Loss{\pi}{s_0}.
\end{equation*}



%

\subsection{Containment Objective}

We also investigate whether an infection can be contained (in expectation). We consider two different questions. For inifinite graphs, we ask: can an infection ever be contained under a finite budget $b$? (i.e., will it grow forever or not?). On finite graphs, we consider containment with thresholds. That is, given a maximal cardinality of infected nodes (i.e., a worst-case ``acceptable'' loss), is a budget $b$ sufficient to ensure an infection does not grow to be larger than that cardinality? The second question applies to infinite graphs as well. Specifically, given a threshold ${\theta \in \mathbb{N} \cup \left\{ \infty \right\} }$, we wish to find a minimal budget such that $\Loss{*}{I_0} \leq \theta$ when $\theta < \infty$, or $\Loss{*}{I_0} < \infty$ when $\theta = \infty$.
Formally we have:
\begin{definition}
\label{def:containment}
Let $\theta \in \mathbb{N} \cup \left\{ \infty \right\}$. Let $\Graph{V}{E}$ be a graph on $n$ nodes, and $s_0 = (I_0,B_0)$ be an initial state.
When $\theta < \infty$, we say that \underline{$b$ contains $I_0$ in $\theta$} if $L_{b}^{*}(I_0) \leq \theta$. When $\theta = \infty$, we say that \underline{$b$ contains $I_0$} if $L_{b}^{*}(I_0) < \infty$. 
\end{definition}
\begin{definition}
We say that $b_\theta$ is a \underline{weak upper bound} if for any $b > b_\theta$, $b$ contains $I_0$ in $\theta$.
We say that $b_\theta$ is a \underline{weak lower bound} if for any $b < b_\theta$, $b$ does not contain $I_0$ in $\theta$.
We say $b_\theta$ is a \underline{tight bound} if $b_\theta$ is both a weak upper and lower bound.
\mbox{We define equivalently for the case of $\theta = \infty$ and denote the containment bound by $b_\infty$}.
\end{definition}

\section{Preliminaries}
\label{preliminariesSection}
For $X_k$ discrete-time $\mathbb{R}$-valued process, we denote $\limsup \limits_{k \to \infty} X_k$ by $X_\infty$ 
We use the notation $\subplus{z}$ for any $z \in \mathbb{R}$ to denote the positive part of $z$.
We denote the set of edges in a graph by $E$ and use the notation $u \leftrightarrow v$ to denote $\left(u,v\right) \in E$. We denote by $\Delta$ the maximal node degree in a given graph.
The number of neighbors of a set of nodes is defined to be the number of connected nodes to that set, and is denoted by $N(A)$ for some $A \subset V$. In a similar way, for a state $s = (I,B)$, its neighborhood is defined as the set of nodes
$$
N(s) = N(I)\backslash B.
$$

Another important measure for sets is their cuts. A cut of a set of nodes in a graph is defined to be the number of edges exiting the set.
For any two sets $A,B$ and a node $v$, the cut from $A$ to $v$ is defined by
$$
\cut{A,v} = \abs{\set{\left(u,v\right) \in E}{ u \in A }}.
$$
The cut from $A$ to $B$ is then defined by
$$
\cut{A,B} = \sum_{v \in B} \cut{A,v}.
$$ 
We can then define the cut of $A$ as the cut from $A$ to $N(A)$, that is,
$$
\cut{A} = \cut{A,N(A)}.
$$
Equivalently, for a state $s = (I,B)$ and node $v \in N(s)$, the cut from $s$ to $v$ is defined by \\
\begin{equation*}
\cut{s,v} = \abs{\set{\left(u,v\right) \in E}{ u \in I }},
\end{equation*}
and the cut of $s$ by
$$
\cut{s} = \cut{I,N(s)} = \sum_{v \in N(s)} \cut{s,v}.
$$

\section{First-Order Policies}
\label{fopSection}
Since finding optimal structural policies on general graphs is NP-Hard, we consider ones that are both tractable and intuitive. We propose a simple policy and prove its optimality on trees.

\begin{definition}
A policy $\pi$ is called a first-order policy if at each time $t$ it only vaccinates a subset of nodes in the immediate neighborhood of $I_t$. That is, $\pi(s) \subseteq N(s)$.
\end{definition}

Under first-order policies, the number of potential states we must consider decreases, simplifying the problem of finding optimal policies. In many cases, though, such policies are sub-optimal. We provide such an example in Figure \ref{fig: FOAPExample}. Consider the graph in Figure \ref{fig: FOAPExample}. Next assume, for clarity of the illustration, that infections proceed deterministically (i.e, $p = 1$). In such a case, it is straightforward that an optimal policy would vaccinate $b$ nodes at level 3 at time $t = 0$, then $b$ nodes at level 3 at time $t = 1$. Any other policy would ensure the infection reaches level 4, thus obtaining a loss that is of order $n$, which may be very large. Characterizing settings where first-order policies fail or succeed is a key contribution of this work.

\begin{figure}
	\centering
	\begin{tikzpicture}
	\tikzset{   
		mstyle/.style={column sep=1cm, row sep=0.3cm,nodes={state},font=\bfseries},
		line/.style={draw,very thick,blue},
	}
	
	\matrix(m)[matrix of nodes,ampersand replacement=\&,mstyle]{
		\&          \&          \& $v_{41}$ \\
		\& $v_{21}$ \& $v_{31}$ \& $v_{42}$ \\
		$I_0$ \& $v_{22}$ \&          \& $v_{43}$ \\
		\& $v_{23}$ \& $v_{32}$ \&  \\
		\&          \&          \& $v_{4n}$ \\
	};
	
	\draw[line](m-3-1) -- (m-2-2);
	\draw[line](m-3-1) -- (m-3-2);
	\draw[line](m-3-1) -- (m-4-2);
	
	\draw[line](m-2-2) -- (m-2-3);
	\draw[line](m-2-2) -- (m-4-3);
	\draw[line](m-3-2) -- (m-2-3);
	\draw[line](m-3-2) -- (m-4-3);
	\draw[line](m-4-2) -- (m-2-3);
	\draw[line](m-4-2) -- (m-4-3);
	
	\draw[line](m-2-3) -- (m-1-4);
	\draw[line](m-2-3) -- (m-2-4);
	\draw[line](m-2-3) -- (m-3-4);
	\draw[line](m-2-3) -- (m-5-4);
	\draw[line](m-4-3) -- (m-1-4);
	\draw[line](m-4-3) -- (m-2-4);
	\draw[line](m-4-3) -- (m-3-4);
	\draw[line](m-4-3) -- (m-5-4);
	\end{tikzpicture}
	\caption{An example of a graph for which any first-order policy is sub-optimal. Suppose $b = 1$ and $p = 1$. An optimal policy will vaccinate nodes $v_{31}$ and $v_{32}$ at iterations $t = 0$ and $t = 1$ respectively, such that $\Loss{*}{I_0} = 4$. Any policy which does not vaccinate these nodes at the first two iterations will ensure a loss that is of order $O(n)$.  Specifically, when $b=1$ and $p=1$, a first-order policy $\pi^F$ will obtain a loss $\Loss{\FOAP}{I_0} = n+3$.  }  
	\label{fig: FOAPExample}
\end{figure}
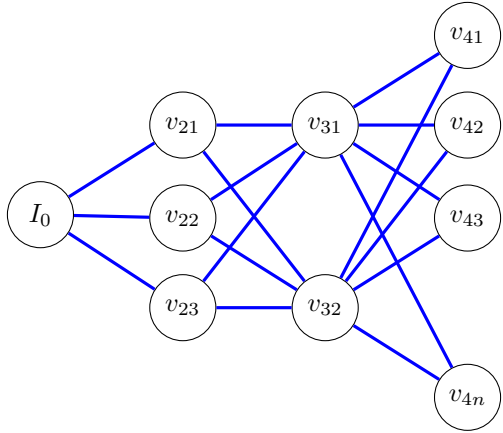

It is interesting to distinguish, when it is optimal to only consider such first-order policies. 
\begin{definition}
	The optimal first-order policy is defined by
	\begin{equation}
	\label{def:FOAP}
	\FOAP(s) \in \argmin \left\{\Loss{\pi}{s} : \pi(s) \subseteq N\left(s\right) \right\} .
	\end{equation}
\end{definition}
\noindent Two natural first-order policies are the CUT policy and the random policy.
\begin{definition}
	The CUT policy is a first-order policy which minimizes the immediate cut, that is
	\begin{equation}
	\label{def:CUT}
	\pi^{CUT}(s) \in \argmin \left\{ \cut{s} \right\}.
	\end{equation}
\end{definition}
\begin{definition}
	The random policy $\pi^R$ is a first-order policy defined by uniformly choosing any $b$ nodes in $N(s)$.
\end{definition}

The CUT and random first-order policies are easy and efficient to calculate. Although they are only an approximation of an optimal first-order policy, their simplicity brings forth the natural question: when is using a first-order policy a good approximation? In the next subsections we will show two instances in which an optimal first-order policy is optimal, and in further sections use first-order policies to achieve bounds for containment, the optimal loss, and the first-order loss. We will ultimately use these containment bounds in Section \ref{stateDependentBudgetSection} together with a first-order policy to achieve a state dependent budget allocation strategy.

\subsection{Optimality of Slow Propagation}

The following theorem is a cornerstone for the rest of this paper.
\begin{theorem}
	\label{FOAPOptimalityThm}
	Let $\Graph{V}{E}$ be a finite graph on $n$ nodes, and let $I_0$ be an initial infection on $G$. Then there exists a positive value $p_{{\rm th}}$ such that $\pi^{F}$ is optimal, for every $p < p_{{\rm th}}$.
\end{theorem}
That is, on finite graphs, there exists a small enough $p$ such that a first-order policy always achieves the smallest loss. The proof reveals, moreover, that a first-order policy may be a good approximation of an optimal policy, as long as $p$ is not too big. Thus, the result of Theorem \ref{FOAPOptimalityThm} underscores the importance of studying fast growing infections as well. Equivalently, one may examine the ratio $\frac{b}{p}$. We would expect problems with a small ratio to be much more simple (i.e., small $p$ and large $b$). A main focus in this paper is studying containment objective which do not confine us to the ``small $p$" or ``large $b$" regimes. \\\\
\begin{proof}
Let ${s = (I,B)}$ be a state. Assume by contradiction there exists a stationary policy $\pi \in \Pi^*$ such that $\pi(s)\cap N(s) \neq \emptyset$. Let $\pi^{F}$ be a first-order policy as defined in equation \ref{def:FOAP}.  Without loss of generality, suppose $s_0 = s$, so that $I_0 = I$. Also, for clarity of the proof, and without loss of generality, assume $B_0 = B = \emptyset$. Throughout this proof we will use $d(A,v)$ to denote the minimal distance between a set $A$ and a node $v$. 
	
	Since $G$ is a finite graph, an absorbing state is reached with probability $1$. We thus use the notation $I_\infty$ to denote $I_{T^*}$, where $T^*$, as defined in Section \ref{optimalityCriteriaSection}, is the random variable depicting the time an absorbing state is reached, that is the time under which $N(s_{T^*}) = \emptyset$. To prove optimality of the first-order policy, we show there exists a small enough $p$, such that vaccinating in $N(I_0)$ ensures a smaller loss. 	
	For a budget $b$ it easy to see that $T^*$ is upper bounded by $\frac{n}{b}$ a.s. For the remainder of the proof we will use the notation $T_a = \frac{n}{b}$ to denote that upper bound.	Let $v \in I_0^c$ with $dist(I_0,v) \geq 2$, and let $M_v$ be the number of paths in $G$ leading from $I_0$ to $v$ of length at most $T_a$. We define the event $A_k^v$ by
	\begin{equation}
	\label{AkDef}
	A_k^v = \{ v \in I_{k} \}.
	\end{equation}
	Since $dist(I_0,v) \geq 2$, it must be that $P(A_1^v) = 0$. Then,
	\begin{equation}
	\label{N2ProbInequality}
	\Pr \giventhat{A_\infty^v}{s_0 = s} \leq 
	\Pr \giventhat{ \bigcup_{k=2}^{\infty} A_k^v }{s_0 = s} =
	\Pr \giventhat{ \bigcup_{k=2}^{T_a} A_k^v }{s_0 = s} \leq
	\sum_{k=2}^{T_a} \Pr \giventhat{ A_k^v }{s_0 = s} \leq
	T_a p^2 M_v,
	\end{equation}
	
	Next we define the event
	\begin{equation*}
	E = \{ I_\infty \subseteq I_0 \cup N(I_0) \},
	\end{equation*}
	and note that $E^c$ can be written in terms of events $A_\infty^u$ from Equation (\ref{AkDef}), as
	\begin{equation*}
	E^c = \bigcup_{u \in \left(I_0 \cup N(I_0)\right)^c} A_\infty^u.
	\end{equation*}
	For brevity, we will use the notation $\tilde{M} = \sum_{u \in \left(I_0 \cup N(I_0)\right)^c} M_u$. Then we have
	\begin{align}
	&\Pr \giventhat{E^c}{s_0} = 
	\Pr \giventhat{\bigcup_{u \in \left(I_0 \cup N(I_0)\right)^c} A_\infty^u}{s_0} \leq \nonumber  \\
	& \leq \sum_{u \in \left(I_0 \cup N(I_0)\right)^c} \Pr \giventhat{ A_\infty^u}{s_0} \leq 
	p^2 \sum_{u \in \left(I_0 \cup N(I_0)\right)^c} T_a M_u = p^2 \tilde{M},
	\label{probBoundEq}
	\end{align}
	where in the first inequality we used the union bound, and in the second we used the inequality in Equation (\ref{N2ProbInequality}).
	 Since $\pi$ does not vaccinate at least one node in $N(s_0)$, the immediate loss of $\pi$ compared to $\pi^F$ can be bounded by
	\begin{equation*}
	\E^{\pi^{F}} \giventhat{\abs{I_{1}}}{s_0} \leq 	
	\E^{\pi} \giventhat{\abs{I_{1}}}{s_0} - p.
	\end{equation*}
	Furthermore, given the event $E$,
	\begin{equation}
	\label{immediateLossEq1}
	\E^{\pi^{F}} \giventhat{\abs{I_{\infty}}}{s_0,E} \leq 	
	\E^{\pi} \giventhat{\abs{I_{\infty}}}{s_0,E} - p.
	\end{equation}
	This is true because given the event $E$ all nodes that are of distance greater than $1$ will never be infected, thus any policy which does not vaccinate at time $t=0$ in $N(s_0)$ will do at least $p$ worse than a first-order policy.
	
	Next note that
	\begin{equation}
	\label{immediateLossEq2}
	\E^{\pi} \giventhat{\abs{I_{\infty}}}{s_0,E} \leq 
	\E^{\pi} \giventhat{\abs{I_{\infty}}}{s_0} =
	\Loss{\pi}{s_0}.
	\end{equation}
	Using Equations (\ref{immediateLossEq1}) and (\ref{immediateLossEq2}) we then obtain
	\begin{equation}
	\label{LossBoundEq1}
	\E^{\pi^{F}} \giventhat{\abs{I_{\infty}}}{s_0,E} \leq \Loss{\pi}{s_0} -p.
	\end{equation}
	Furthermore it trivially holds that
	\begin{equation}
	\label{LossBoundEq2}
	\E^{\pi^{F}} \giventhat{\abs{I_{\infty}}}{s_0,E^c} \leq n.
	\end{equation}
	
	Finally let $p_{th} = \left(n\tilde{M}\right)^{-1}$, and take $p < p_{th}$.
	Then
	\begin{align*}
	&\Loss{\pi^F}{s_0} =
	\E^{\pi^F} \giventhat{\abs{I_{\infty}}}{s_0} =
	\E^{\pi^F} \giventhat{\abs{I_{\infty}}}{s_0,E} \Pr \giventhat{E}{s_0} +
	\E^{\pi^F} \giventhat{\abs{I_{\infty}}}{s_0,E^c} \Pr \giventhat{E^c}{s_0} \leq \\
	& \leq \left(\Loss{\pi}{s_0}-p\right) \Pr \giventhat{E}{s_0} + 
	n \Pr \giventhat{E^c}{s_0} \leq
	\Loss{\pi}{s_0}-p + np^2 \tilde{M} < \Loss{\pi}{s_0},
	\end{align*}
	where here the first inequality uses Equations (\ref{LossBoundEq1}) and (\ref{LossBoundEq2}). In the second inequality we have used the fact that $\Pr \giventhat{E}{s_0} \leq 1$ as well as the inequality in equation \ref{probBoundEq}. Finally, the last strict inequality uses the fact that $p < p_{th} = \left(n\tilde{M}\right)^{-1}$. \\
	We have reached a contradiction to the optimality of $\pi$, thereby completing the proof.
\end{proof}

\subsection{Optimality on Trees}
	\label{TreeExampleSection}	
	We show next that regardless of the transmission probability $p$, first-order policies are optimal on trees.
	
	\begin{theorem}
		\label{optimalTreeThm}
		Let $T$ be a tree, and suppose an infection initiates at its root. Then $\pi^F$ is optimal.
	\end{theorem}
\begin{proof}
Let $T_x$ denote the subtree of a node $x$ in the tree, and note that 
		\begin{equation}
		\label{eq:TreeOptimality1}
		x \in B_t \Rightarrow T_x \cap I_\infty = \emptyset.
		\end{equation}
		Assume by contradiction that the first-order policy, $\pi^F$, is sub-optimal. Let $\pi$ be an optimal policy that is not a first-order policy. Then there is a state from which $\pi$ vaccinates a node outside the neighborhood of $s_t$. Denote this node by $u$. Then, there exists a node $v \in N(s_t)$ such that $v$ is in a shortest path between $I_t$ and $u$, which is also a unique path, since $T$ is a tree. Also, we have $T_u \subset T_v$. Let $\tilde{\pi}$ be the policy which vaccinates $v$ instead of $u$ at time $t$. Then by these facts, and using Equation (\ref{eq:TreeOptimality1}),
		\begin{equation*}
		I_\infty^\pi \geq I_\infty^{\tilde{\pi}} \quad a.s.
		\end{equation*}
		We can continue this procedure until $\tilde{\pi}(s_t) \subseteq N(s_t)$. Therefore, $\tilde{\pi}$ is a first-order policy, and achieves a minimal loss, which contradicts our assumption that $\pi^F$ is sub-optimal. 
\end{proof}

\subsection{First-Order Policy for Trees}
	Theorem \ref{optimalTreeThm} shows us that for trees, a first-order policy is always optimal. Finding this optimal first-order policy, however, can still be NP-Hard (see \cite{finbow2007firefighter}). We therefore look for instances of trees where an explicit structural policy can be established. For this reason, we consider the case of $d$-regular trees.
	 
	Let $T$ be a regular tree of degree $d$. We denote by $Lev(v)$ the level (from the root) of a node $v$, where we consider the root to be of level 1, that is, $Lev(root) = 1$. We also denote the number of levels by $L$, where we allow $L = \infty$. We assume an initial infection $I_0$ spreads in $T$ such that $I_0$ is a connected set. We define the set $R_t$ as
	\begin{equation*} \label{eq1}
	R_t = 					   
	\argmin_{u \in N(s_t)} Lev(u).
	\end{equation*}
	
	\noindent Note that once $root \notin I_0$ it must be that $\abs{R_0} = 1$.
	
	In order to obtain exact analytical results, we disregard edge effects by assuming $T$ is deep enough such that the infection never reaches its leaves. That is, we implicitly assume a limit of $L \to \infty$. Algorithm \ref{alg:treePolicy} outlines an optimal structural policy for regular trees, which is indeed a first-order policy since it only considers nodes in $N(I_t)$. We prove this policy is indeed optimal in the following Theorem.
	
\begin{algorithm}[t]
	\SetAlgoNoLine
	\KwIn{State $s_t$.}
	\KwOut{Next node to vaccinate $v$}
	\eIf {$R_0 = root$ and $\abs{N(s_t)} \geq 2$}{
	$v \in \argmin_{u \in N(s_t) \backslash R_0} Lev\left(u\right)$}{
$v \in R_t$
	}
	\Return $v$
\caption{Tree Policy}
\label{alg:treePolicy}
\end{algorithm}
	
	\begin{theorem}
		\label{optimalTreePolicyThm}
		Let $T$ be a regular tree, and let $I_0$ be a connected set. Then the tree policy outlined in Algorithm \ref{alg:treePolicy} is optimal.
	\end{theorem}
\begin{proof}
Recall that throughout all of our derivation we assume $I_0$ is a connected set, so that $I_t$ must also be a connected set for all $t \geq 0$.  \\
		
		\noindent \textbf{Case 1: $root \in I_0$} \\
		
		Denote by $T_i$ a subtree of $T$ with its root at level $i$ of $T$. Also denote by $L^*_{T_i}$ the optimal loss for an infection spreading on the graph $G_i = T_i$ with an infection initiated at its root. 
		It follows that
		\begin{equation}
		\label{treeInequality}
		L^*_{G_i} \leq L^*_{G_j}, \forall i \geq j.
		\end{equation}

		Let $s = (I,B)$ with $root \in I$. Using the fact that $\FOAP(s) \subseteq N(s)$ together with equation \ref{treeInequality} yields $\FOAP(s) \in \argmin_{u \in N(s)} Lev\left(u\right)$. \\
		
		\noindent \textbf{Case 2: $root \notin I_0$ and $R_0 \neq root$} \\	
		
		We create a new tree $T'$ as follows. Take all nodes in $I_0$ and merge them together to one node $v_0$ such that all neighbors of $I_0$ remain neighbors of $v_0$. $T'$ is an asymmetric tree with $v_0$ as its root. One of the children of $v_0$ is $R_0$. Since $R_0 \neq root$ it has exactly $d$ healthy children in $T'$. This means we can use equation \ref{treeInequality} of Case 1 for $T'$ (as the root is contained in $v_0 = I_0$). Hence, since $R_0$ is of lower level than all other children of $v_0$, it must be vaccinated first.\\
		
		\noindent \textbf{Case 3: $root \notin I_0$ and $R_0 = root$} \\

		In this case, we again construct $T'$ as in case 2. Once $T'$ is created, $R_0$ will have $d-1$ healthy children, while all other children of $v_0$ will have $d$ children. This comes into conflict with Equation (\ref{treeInequality}) which deals with a case of all nodes having the same degree. 
		
		Since we only consider the case of $L \to \infty$, an optimal policy must choose to vaccinate nodes with higher degree, that is any nodes in the neighborhood of $v_0$ excluding $R_0$, leaving $R_0$ as a last priority for vaccination. A visualization of the process of building $T'$ for this case is depicted in Figure \ref{markerTtag}. 
		
		\begin{figure}[t!]
			\centering
			\begin{subfigure}{0.3\textwidth}
				\centering
				\includegraphics[width=0.9\linewidth]{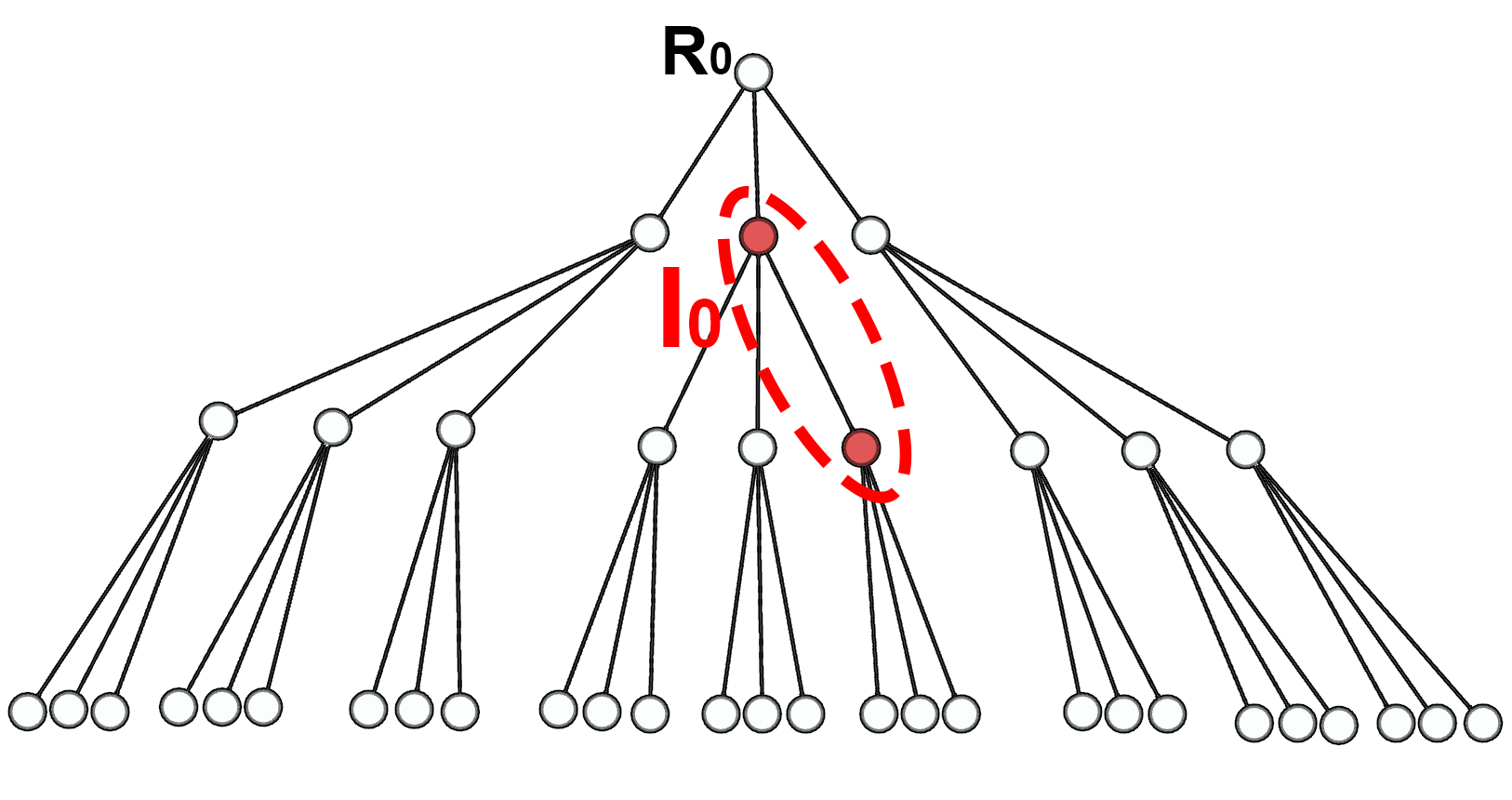}
				\caption{}
			\end{subfigure}
			\begin{subfigure}{0.3\textwidth}
				\centering
				\includegraphics[width=0.9\linewidth]{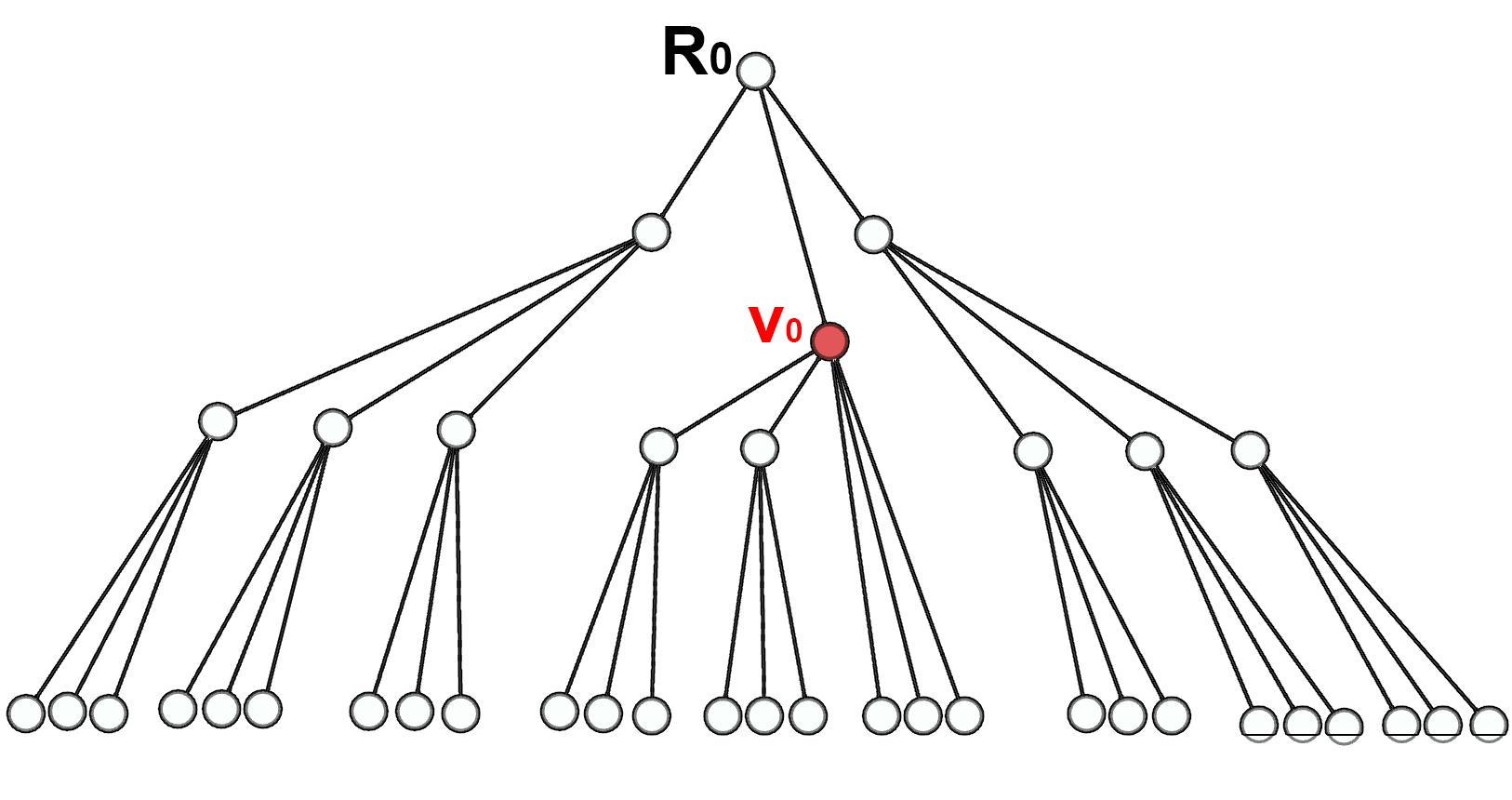}
				\caption{}
			\end{subfigure}
			\begin{subfigure}{0.3\textwidth}
				\centering
				\includegraphics[width=0.9\linewidth]{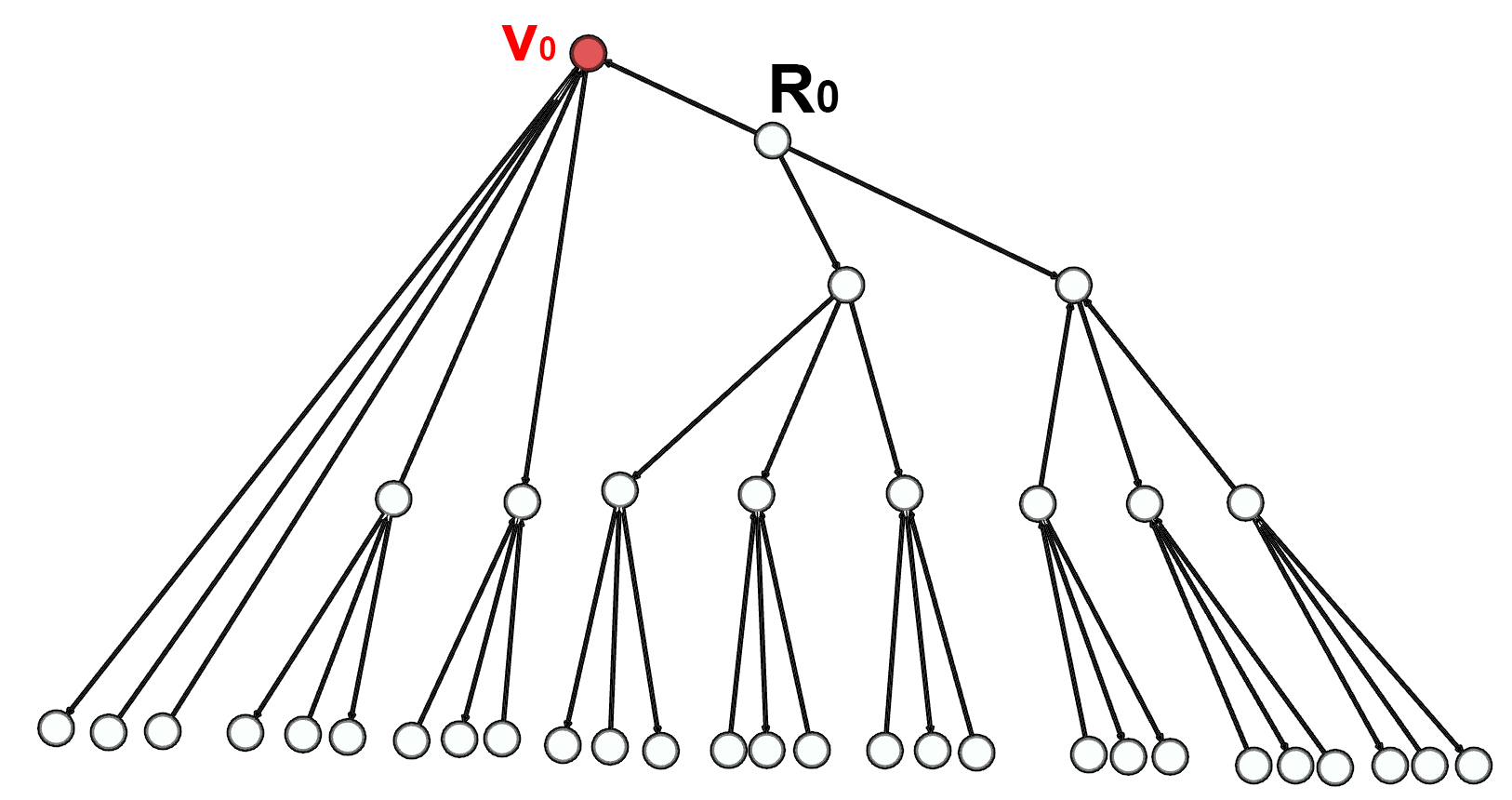}
				\caption{}
			\end{subfigure}
			\caption{ The creation of $T'$ for $R_0 = root$ and $d = 3$. (a) Original tree $T$. (b) Merge $I_0$ to one node $v_0$. (c) Visualization of new tree with $v_0$ as its root. }
			\label{markerTtag}
		\end{figure}
\end{proof}
	
Theorems \ref{optimalTreeThm} and \ref{optimalTreePolicyThm} give us optimal policies on trees, but do not tell us the actual expected loss, nor do they give us estimates on the needed budget for containment. In order to obtain bounds on these quantities, and also to study general (non-tree) topologies, we develop the concept of growth rates in the next section.

\section{Growth Rate} \label{growthRateSection}
\label{GrowthRateSection}
This section develops the required machinery for developing upper and lower bounds on the optimal loss. These bounds enable us to answer questions regarding containment objective. We define the notion of the growth rate. Informally, the growth rate of an infection at time $t$ is the expected cardinality growth of $I_t$ between time $t$ and $t+1$. That is, if we denote by $GR_t$ the growth rate at time $t$, then
\begin{equation}
\label{eq:gr}
\E \giventhat{\abs{I_{t+1}}}{s_t} = \abs{I_t} + GR_t.
\end{equation}
 The growth rate thus changes at each time step $t$. Intuitively, a good policy would (1) maintain a small growth rate throughout the propagation of the infection and (2) bring the growth rate to zero as quickly as possible. 
 We define the notion of the growth rate formally below.
 \begin{definition}
 	The growth rate of a set $A$ to a set $B$ for infection speed $p$ is defined by
 	\begin{equation}
 	\label{eq:gr1}
 	\GR{A,B}{p} = \sum_{v \in B} \left(1-(1-p)^{\cut{A,v}}\right).
 	\end{equation}
 	The growth rate of a set $A$ is defined using equation (\ref{eq:gr1}) as 
 	$
 	\GR{A}{p} = \GR{A,N(A)}{p}
 	$.
 	Equivalently, we define the growth rate of a state $s = (I,B)$ to be
 	$
 	\GR{s}{p} = \GR{I,N(s)}{p}
 	$.
 \end{definition}
\noindent In order to get a better understanding of the growth rate, we note the following simple lemma, which states that the number of neighbors and cut of a state behave as lower and upper bounds of the growth rate of that state, respectively.
\begin{lemma}
	\label{GRProperty}
	Let $s = (I,B)$ be a state. Then
	$\abs{N(I)\cap B}p \leq \GR{I,B}{p} \leq \min \left\{ \cut{I,B}p,1 \right\}$.
\end{lemma}
\begin{proof}
For any $p \in \left[0, 1\right]$ and $k \in \mathbb{N}$, 
		\begin{equation*}
		1 - (1-p)^k \leq \min \left\{ pk, 1\right\}.
		\end{equation*}
		Then
		\begin{equation*}
		\GR{I,B}{p} =  \sum_{v \in B} \left(1-(1-p)^{\cut{I,v}}\right) \leq  \sum_{v \in B} \min \left\{ p\cut{I,v}, 1\right\}  = \min \left\{ \cut{I,B} p, 1 \right\}.
		\end{equation*}
		
		On the other hand, for any $v \in N(I)$,
		\begin{equation*}
		\cut{I,v} \geq 1.
		\end{equation*}
		Also, if $v \notin N(I)$ then $\cut{I,v} = 0$.
		Then for any $p \in \left[0, 1\right]$,
		\begin{align*}
		&\GR{I,B}{p} = \sum_{v \in B} \left(1-(1-p)^{\cut{I,v}}\right) = \sum_{v \in B\cap N(I)} \left(1-(1-p)^{\cut{I,v}}\right)  \\
		&\geq \sum_{v \in B\cap N(I)} \left(1-(1-p)\right) = \abs{B\cap N(I)} p.
		\end{align*} 
\end{proof}
 \noindent The growth rate of a state $s = (I,B)$, $\GR{s}{p}$, takes into account the vaccinated set of nodes $B$. It can be useful to write it using the growth rate for sets in the following way:
 \begin{align}
 \label{GRSeparationEq}
 &\GR{s}{p} = \sum_{v \in N(s)} \left(1-(1-p)^{\cut{I,v}}\right) = \sum_{v \in N(I)\backslash B} \left(1-(1-p)^{\cut{I,v}}\right)  \nonumber \\
 &= \sum_{v \in N(I)} \left(1-(1-p)^{\cut{I,v}}\right) - \sum_{v \in B} \left(1-(1-p)^{\cut{I,v}}\right) = \GR{I}{p} - \GR{I,B}{p}.
 \end{align}
 
\subsection{Maximal and Minimal Growth Rates over Upward Crusades}
\label{UpwardCrusadesSection}

\begin{figure}[t!]
\centering
\begin{subfigure}{0.2\textwidth}
	\centering
	\includegraphics[width=0.9\linewidth]{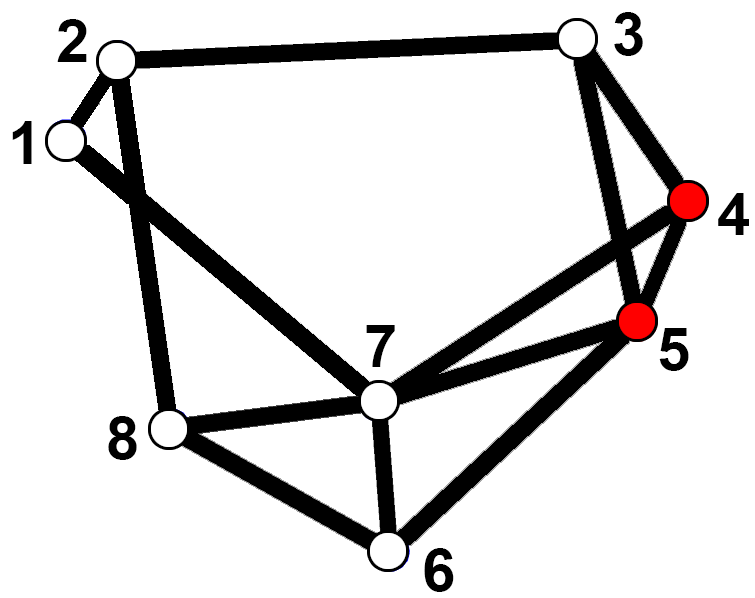}
	\caption{}
\end{subfigure}
\begin{subfigure}{0.2\textwidth}
	\centering
	\includegraphics[width=0.9\linewidth]{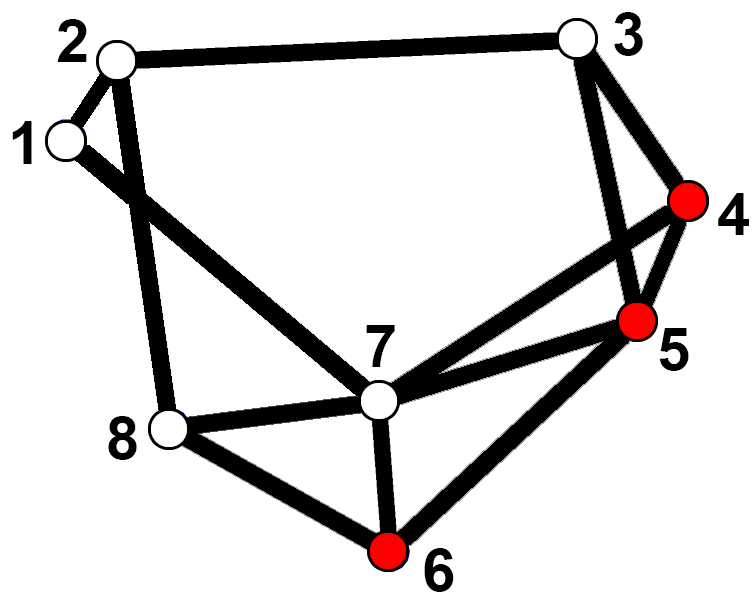}
	\caption{}
\end{subfigure}
\begin{subfigure}{0.2\textwidth}
	\centering
	\includegraphics[width=0.9\linewidth]{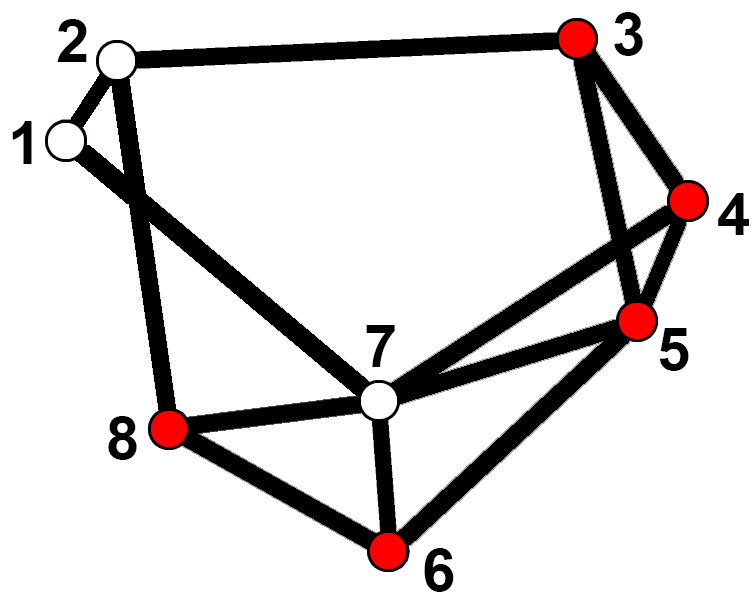}
	\caption{}
\end{subfigure}
\begin{subfigure}{0.2\textwidth}
	\centering
	\includegraphics[width=0.9\linewidth]{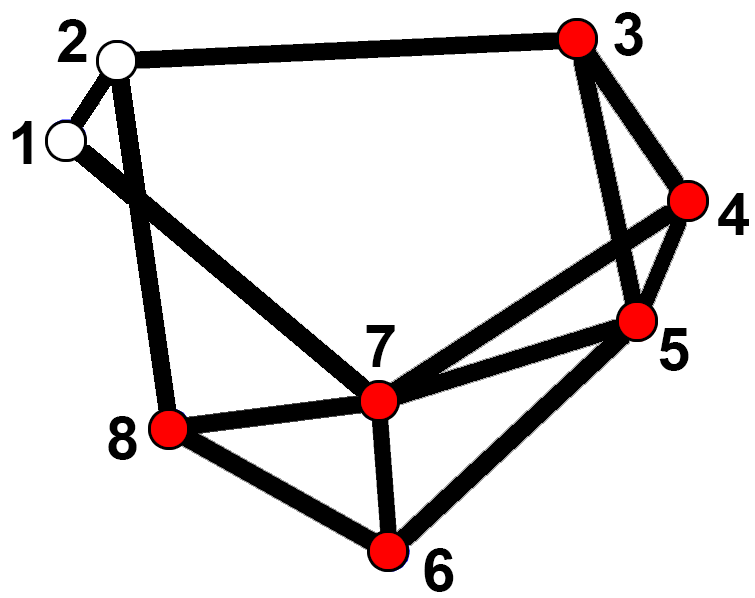}
	\caption{}
\end{subfigure}
\begin{subfigure}{0.2\textwidth}
	\centering
	\includegraphics[width=0.9\linewidth]{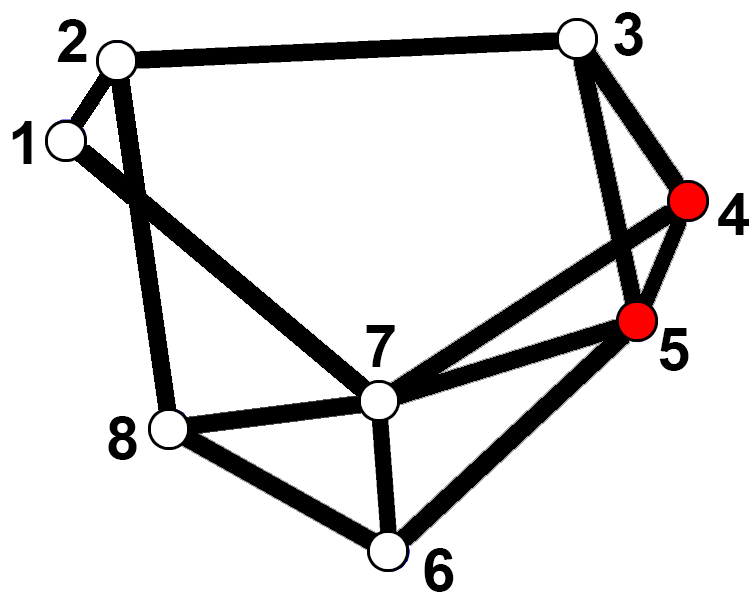}
	\caption{}
\end{subfigure}
\begin{subfigure}{0.2\textwidth}
	\centering
	\includegraphics[width=0.9\linewidth]{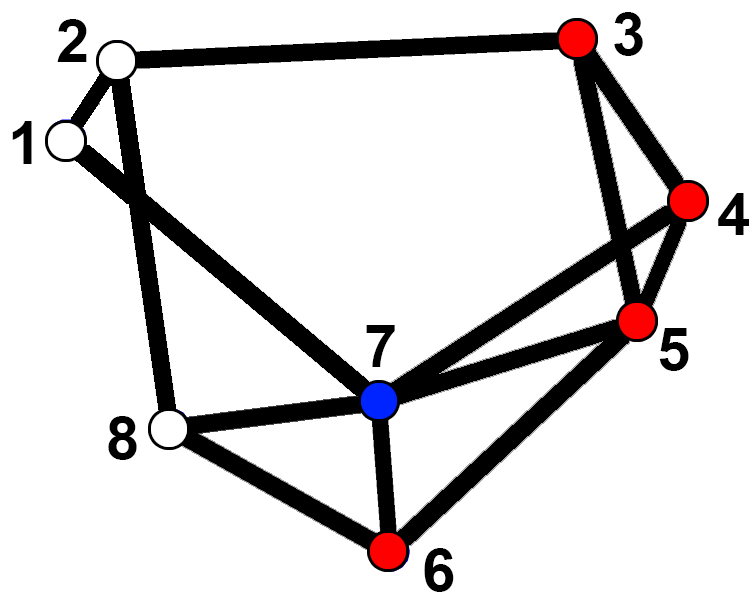}
	\caption{}
\end{subfigure}
\begin{subfigure}{0.2\textwidth}
	\centering
	\includegraphics[width=0.9\linewidth]{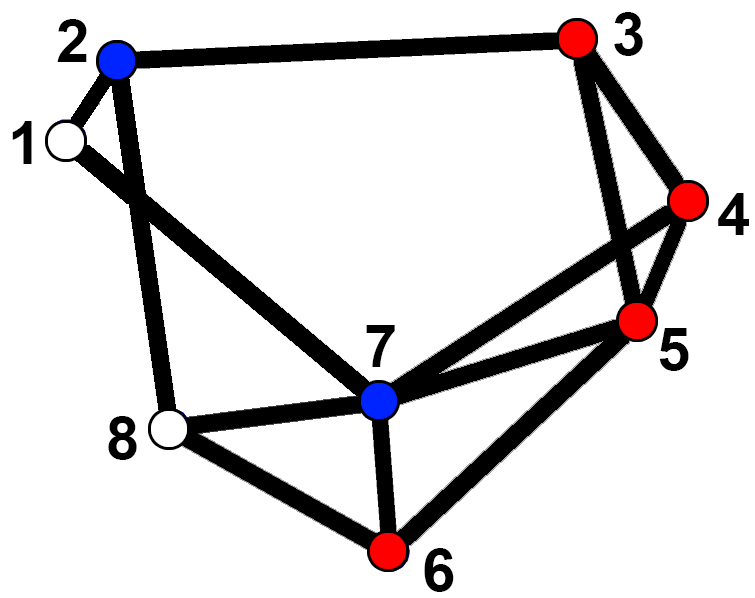}
	\caption{}
\end{subfigure}
\begin{subfigure}{0.2\textwidth}
	\centering
	\includegraphics[width=0.9\linewidth]{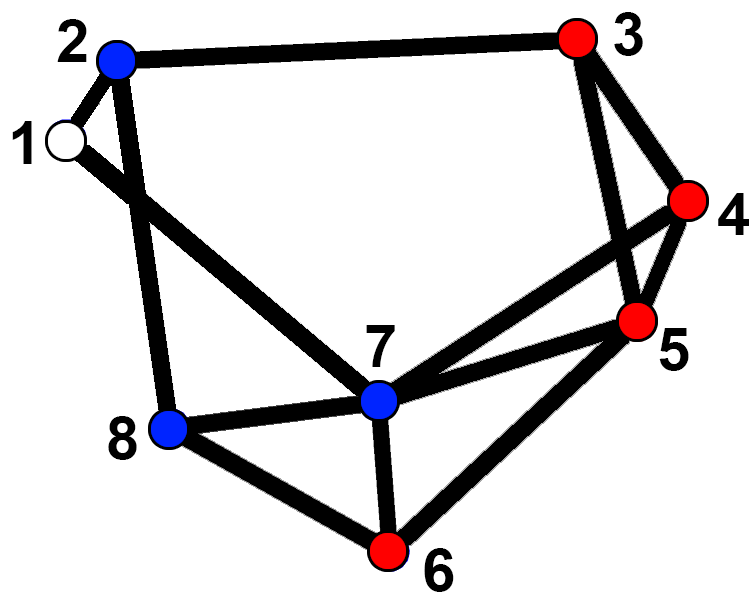}
	\caption{}
\end{subfigure}
\caption{ (a-d) An example of an upward crusade $u \in \UC{A}{3}$, where $A = \left\{ 4,5\right\}$. The sequence depicted is ${u = \left(  \left(\left\{ 4,5\right\},\emptyset \right), \left(\left\{ 4,5,6\right\},\emptyset \right),  \left(\left\{ 3,4,5,6,8\right\},\emptyset \right),  \left(\left\{ 3,4,5,6,7,8\right\},\emptyset \right)\right)}$. Throughout the whole sequence, $b=0$ and $B_t = \emptyset$.  (e-h) An example of an upward crusade $u \in \UC{\pi,s}{3}$ of some policy $\pi \in \Pi$ under a budget of $b = 1$, and $s = (I,\emptyset)$, where $I = \left\{ 4,5\right\}$. The sequence depicted is ${u = \left(  \left(\left\{ 4,5\right\}, \left\{\emptyset\right\} \right), \left(\left\{ 3,4,5,6\right\}, \left\{7 \right\} \right),\left(\left\{ 3,4,5,6\right\}, \left\{2,7 \right\} \right),\left(\left\{ 3,4,5,6\right\}, \left\{2,7,8 \right\} \right) \right)}$.  Note that between states $s_2$ and $s_3$, the infected set does not change.}
\label{markerTtag}
\end{figure}

Growth rates give us a tool for measuring the expected growth of an infection at a time $t$, but do not take into account the growth rate in future states. In order to build better structural policies, we consider the full propagation of the epidemic using the notion of maximal and minimal growth rates. These are merely parameters that tell us how large or how small the growth rate can be under a specific policy in a given time window. To define these, we follow \cite{drakopoulos2014efficient}, and define upward crusades as sequences of possible states under a given policy. We then define the maximal and minimal growth rates as the worst and best cases of growth rates over all possible sequences. 
 
An upward crusade is a sequence of infection states which follow all possible realizations of infections starting at some initial state $s_0$ under some policy $\Pi_{Sb}$. Upward crusades are important as they enable us to define maximal and minimal growth rates, which then enable us to obtain bounds on the optimal loss. 

\begin{definition}
For a state $s = (I,B)$ and policy $\Pi_{Sb}$, an upward crusade of length $k$ is a sequence $u = \left(s_0,\hdots,s_k\right)$ of $k+1$ pairs (states), $s_i = \left( I_i, B_i \right)$, with the following properties: 
\begin{enumerate}
	\item $s_0 = s$
	\item $I_i \backslash I_{i-1} \subseteq N(I_{i-1})\backslash B_i$, for $i = 1\hdots k$
	\item $B_i \backslash B_{i-1} = \pi(s_{i-1})$, for $i = 1\hdots k$
\end{enumerate}
We denote the final state of the sequence by $u_k$. We also denote the set of all upward crusades of length $k$ initiating at a state $s$ by $\UC{\pi,s}{k}$.
Finally, we denote the set of all possible upward crusades of length $k$ for the set of policies $\Pi_0$ by $\UC{\Pi_0,s}{k}$, that is
\begin{equation*}
\UC{\Pi_0,s}{k} = \bigcup_{\pi \in \Pi_0} \UC{\pi,s}{k}.
\end{equation*}
\end{definition}
It is sometimes helpful to consider upward crusades that do not vaccinate any nodes. These upward crusades depict worst case instances of infection states in the graph under arbitrary policies. When we wish to consider such upward crusades, we assume $b = 0$, and denote them by $\UC{I}{k}$, where here we omit $\pi$ (since $b = 0$), and use $I$ to denote the initial set of infected nodes. Figure \ref{markerTtag} depicts examples of upward crusades on a simple graph.

\begin{remark}[Inclusion of upward crusades]
\label{UCInclusionRemark}
Note that for any state $s = (I_0,B_0)$ and $u \in \UC{\pi,s}{k}$, with $u_i = (I_i,B_i)$ there exists $w \in \UC{I_0}{k}$ such that $w=(I_0,I_1, \hdots, I_k)$.
\end{remark}
\noindent The following lemma is an important monotonicity property of upward crusades.
\begin{lemma}[Monotonicity in $k$]
Suppose $k_2 \geq k_1$. Let $u \in \UC{I}{k_1}$. Then there exists $w \in \UC{I}{k_2}$ 

such that $u \subseteq w$ and $u_{k_1} = w_{k_2}$.
\end{lemma}
\begin{proof}
Let $u \in \UC{I}{k_1}$. We build an upward crusade $w \in \UC{I}{k_2}$ as follows. The first $k_1 +1$ sets of the sequence $w$ are the same as $u$. In the final $k_2-k_1$ sets we follow the upward crusade satisfying $I_i \backslash I_{i-1} = \emptyset$. This in turn results in $u_{k_1} = w_{k_2}$. 
\end{proof}
Maximal and minimal growth rates let us bound the expected loss of a policy. 
Using upward crusades, the maximal growth rate looks $k$ steps into the future, until the point when an infection of cardinality $c$ is reached under some policy $\Pi_{Sb}$. It then returns the worst case growth rate over all corresponding upward crusades. This enables us to bound the loss of the policy, as will be shown in Theorem \ref{MGRThm}. A similar idea follows for the minimal growth rate. We define these formally in the following definition.

\begin{definition} [Maximal and Minimal Growth Rates]
	Given a policy $\pi \in \Pi$, a state $s$, and the set 	of upward crusades of length $k$, $\UC{\pi,s}{k}$, the maximal and minimal growth rates are functions ${MGR, mgr: \left\{ 1, \hdots, \abs{V} \right\} \times \UC{\Pi,s}{k} \to \mathbb{R}_+}$, 
	defined by
	\begin{align*}
	&\MGR{c,\UC{\pi,s}{k}}{p} = \maxl_{\substack{u \in \UC{\pi,s}{k} \\ u_k \in S_{c} }} \GR{u_k}{p}, \text{ and} \\
	&\mgr{c,\UC{\pi,s}{k}}{p} = \minl_{\substack{u \in \UC{\pi,s}{k} \\ u_k \in S_{c} }} \GR{u_k}{p}, 
	\end{align*}
	
	\noindent where $S_{c} = \set{s = (I,B)}{ \abs{I} = c}$, and recall that $u_k$ denotes the final state in an upward crusade $u$. In other words, the maximum / minimum growth rates maximize the growth rate over all states that end upward crusades in $\UC{\pi,s}{k}$ and have infection cardinality $c$. 
	\end{definition}
	\begin{definition} [Expected Growth Rate]
	Given a policy $\pi \in \Pi$, a state $s$, and the set 	of upward crusades of length $k$, $\UC{\pi,s}{k}$, the expected growth rate is a function	
	${EGR: \UC{\Pi,s}{k} \to \mathbb{R}_+}$, 
	defined by
\begin{equation*}
	\EGR{c, \UC{\pi,s}{k}}{p} = \E^\pi \giventhat{ \GR{s_k}{p} } { \abs{I_k} = c, s_0 = s}  = \int_{\abs{I_k} = c} \GR{u_k}{p} dF_{u_k},
	\end{equation*}
	\noindent \mbox{where $u_k = (I_k, B_k)$ is the random variable of the final state in upward crusades of $\UC{\pi,s}{k}$}, such that the integral is taken over final states with cardinality $\abs{I_k} = c$.
	\end{definition}
	It is helpful to consider maximal, minimal, and expected growth rates of an empty policy which does not vaccinate any nodes. We will denote these using upward crusades $\UC{I}{k}$ in place of $\UC{\pi,s}{k}$ to emphasize the fact the maximum/minimum/expectation is taken over crusades which do not vaccinate any nodes (i.e., $b = 0$). This notation is useful for finding bounds for specific topologies (see Section \ref{containmentSection}).
%
%
%


\begin{figure}[t!]
\centering
\begin{subfigure}{0.4\textwidth}
	\centering
	\includegraphics[width=0.9\linewidth]{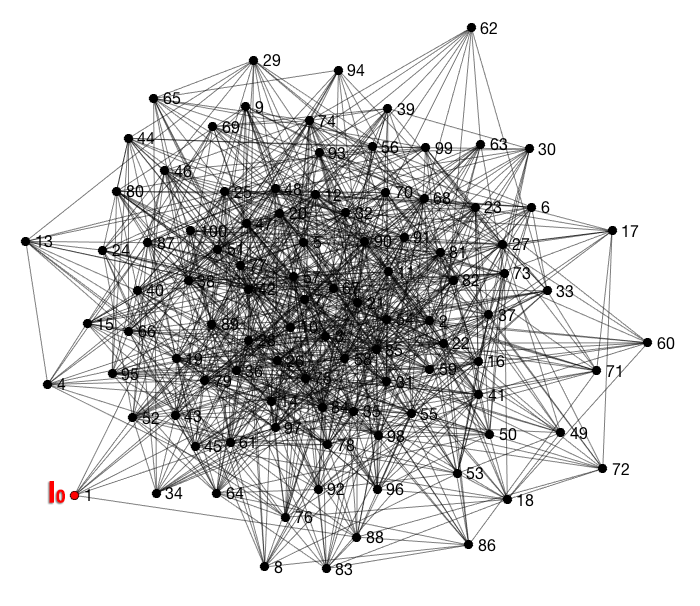}
	\caption{}
\end{subfigure}
\begin{subfigure}{0.4\textwidth}
	\centering
	\includegraphics[width=0.9\linewidth]{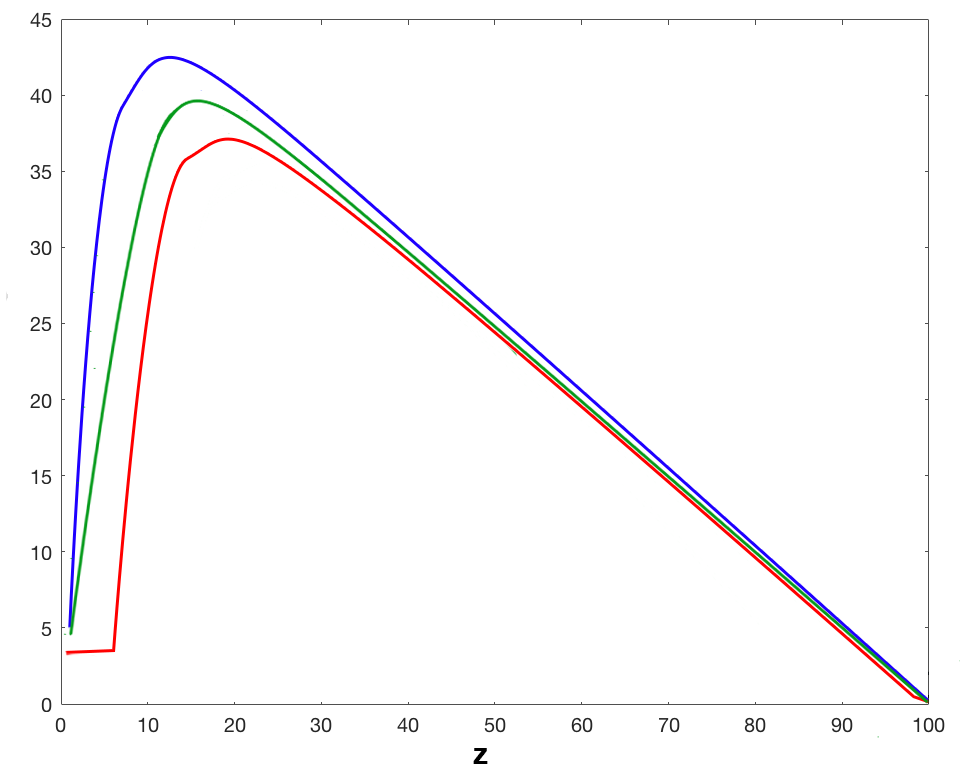}
	\caption{}
\end{subfigure}
\caption{ (a) A sample graph of an Erd\H{o}s-R\'{e}nyi model $G(n,s)$, with parameters $n = 100, s = 0.1$ with diameter 3. \\
(b) Plots of $\MGR{z,\UC{I_0}{100}}{\frac{1}{2}}$ (in blue), $\mgr{z,\UC{I_0}{100}}{\frac{1}{2}}$ (in red), and $\EGR{z,\UC{I_0}{100}}{\frac{1}{2}}$ (in green) as functions of $z$, with $I_0 = \left\{ 1 \right\}$. Estimates were made using Monte Carlo simulations (see Section \ref{stateDependentBudgetSection} for details). }
\end{figure}

\newpage
\subsection{Upper and Lower Bounds}
\label{BoundsSection}

Calculating the loss of a policy as well as the optimal loss can be done using Monte Carlo simulations. Finding analytical expressions, though, is in most cases not possible. To overcome this problem, we look for a method of bounding the loss of a policy. Theorem \ref{MGRThm} explicitly calculates upper and lower bounds for the loss. These bounds allow us to obtain specific conditions for containment, as we explain in Section \ref{containmentSection}.

Let $\Graph{V}{E}$ be a graph on $n$ nodes, and let $s_0$ be an initial state. We wish to bound from above the loss of a policy $\pi \in \Pi$. We define a process $M_t$ as the expected cardinality of the infection at time $t$ given the initial state and the policy, that is
\begin{equation*}
M_t^\pi =  \E^{\pi} \giventhat{\abs{I_t}}{s_0}.
\end{equation*}
$M_t^\pi$ is directly linked to the loss of a policy $\pi$. The loss of a policy $\pi$ can be written in terms of $M_t^\pi$ as $L^{\pi}(s_0) = M_{T^*}^\pi$.  Equation (\ref{eq:gr}) gives an interpretation of the growth rate which we can leverage to calculate $M_t^\pi$ and consequently obtain $\Loss{\pi}{s_0}$. This calculation is in general intractable. Thus, we turn to maximal, minimal, and expected growth rates to bound $M_t^\pi$ from below or above (Section \ref{stateDependentBudgetSection} also gives a method of approximating these bounds). Our main approach includes defining recursion relations which lower/upper bound the process $M_t^\pi$. We define the functions
\begin{align}
&\phi^\pi_1(c,k) = \MGR{c,\UC{\pi,s_0}{k}}{p}, \label{MGR} \\
&\phi^\pi_2(c,k) = \mgr{c,\UC{\pi,s_0}{k}}{p} \label{mgr} \text{, and } \\
&\phi^\pi_3(c,k) = \EGR{c,\UC{\pi,s_0}{k}}{p}. \label{EGR}
\end{align}
We use linear interpolation on the cardinality parameter $c$, and define the following recursion relations ($1 \leq i \leq 3$):
\begin{align}
&X_i(k+1) = X_i(k) + \phi^\pi_i\pth{X_i(k),k} &,k \geq 0, \abs{I_0} \leq X_i(k) \leq \theta  \nonumber \\ 
&X_i(0) = \abs{I_0}, \label{recursion}
\end{align}
where $\theta \in \mathbb{N} \cup \left\{ \infty \right\}$ is a given parameter.

We denote by $k_i$ the maximal $k$ in which the conditions of recursion (\ref{recursion}) are satisfied, namely, ${k_i = \max \set{k \geq 0}{X_i(k) \leq \theta}}$. The following theorem gives general bounds on the process $M_t^\pi$.

\begin{theorem}[Bounds]
\label{MGRThm}
Let $M_t$, $X_i(k)$ be defined above.
\begin{enumerate}
	\item Suppose $\phi^\pi_1(z,k)$ is concave and monotone non-decreasing in $z$, ${\forall (z,k) \in \left[ \abs{I_0} ,\theta \right] \times \left\{ 0, \hdots, k_1 \right\} }$. Then $M_t^\pi \leq X_1(t)$ for any $t \in \left\{ 0, \hdots, k_1 \right\}$.
	\item Suppose $\phi^\pi_2(z,k)$ is convex and monotone non-decreasing in $z$, ${\forall (z,k) \in \left[ \abs{I_0} ,\theta \right] \times \left\{ 0, \hdots, k_2 \right\}}$. 	Then $M_t^\pi \geq X_2(t)$ for any $t \in \left\{ 0, \hdots, k_2 \right\}$.
		\item Suppose $\phi^\pi_3(z,k)$ is concave (convex) and monotone non-decreasing in $z$, \\
		${\forall (z,k) \in \left[ \abs{I_0} ,\theta \right] \times \left\{ 0, \hdots, k_3 \right\} }$. \\
		Then $M_t^\pi \leq X_3(t)$ ($M_t^\pi \geq X_3(t)$) for any $t \in \left\{ 0, \hdots, k_3 \right\}$.
\end{enumerate}
\end{theorem}
\begin{proof}
We prove by induction on  $t \in \left\{ 0, \hdots, k_i \right\}$. \\
	For $t = 0$:
	\begin{equation*}
	M_0^\pi = \E^{\pi} \giventhat{\abs{I_0}}{s_0} = \abs{I_0} = X_0.
	\end{equation*}
\begin{enumerate}
	\item 
	Assume that $M_k \leq X_k$ for some $t = k \leq k_1-1 $. By definition of $\phi^\pi_1$, the maximal growth rate function, and for any state $s_k \in u$ with $u \in \UC{\pi,s_0}{k}$,
	\begin{equation*}
	\E^{\pi} \giventhat{\abs{I_{k+1}}}{s_k} = \abs{I_k} + \GR{s_k}{p} \leq \abs{I_k} + \MGR{\abs{I_k},\UC{\pi,s_0}{k}}{p} = \abs{I_k} + \phi^\pi_1(\abs{I_k},k).
	\end{equation*}
	
	\noindent Next we bound $M_{k+1}$ by
	\begin{align*}
	&M_{k+1}^\pi = \E^{\pi} \giventhat{\abs{I_{k+1}}}{s_0} = \E^{\pi} \giventhat{\E^{\pi}  \giventhat{\abs{I_{k+1}}}{s_k,s_0}}{s_0} \leq 
	\E^{\pi} \giventhat{\abs{I_k} + \phi^\pi_1(\abs{I_k},k)}{s_0} = \\
	& = M_k^\pi + \E^{\pi} \giventhat{\phi^\pi_1(\abs{I_k},k)}{s_0} \stackrel{\text{(a)}}{\leq} M_k^\pi + \phi^\pi_1 \left(\E^{\pi} \giventhat{\abs{I_k}}{s_0},k \right) =
	M_k^\pi + \phi^\pi_1 \left(M_k^\pi,k \right) \stackrel{\text{(b)}}{\leq}\\
	&\stackrel{\text{(b)}}{\leq}  X_k + \phi^\pi_1 \left(X_k,k \right) = X_{k+1},
	\end{align*}
	where in inequality (a) we used concavity of $\phi^\pi_1(z,k)$ in $z$ for $k \in \left\{ 0, \hdots, k_1 \right\}$, and in inequality (b) monotonicity of $\phi^\pi_1(z,k)$ in $z$ for $k \in \left\{ 0, \hdots, k_1 \right\}$ as well as the induction step.
	\item The proof follows the same steps as in 1.	
	\item By definition of $\phi^\pi_3$, the expected growth rate function, and for any state $s_k \in u$ with $u \in \UC{\pi,s_0}{k}$,
	\begin{equation*}
	\E^{\pi} \giventhat{\abs{I_{k+1}}}{\abs{I_k}, s_0} = \abs{I_k} + \E^{\pi} \giventhat{\GR{s_k}{p}}{\abs{I_k} = z, s_0} = \abs{I_k} + \phi^\pi_3(\abs{I_k},k).
	\end{equation*}
	If $\phi^\pi_3$ is convex then assume that $M_k \leq X_k$ for some $t = k \leq k_3-1 $, then using the induction step and monotonicity of $\phi^\pi_3$ the result is obtained as in 1. If $\phi^\pi_3$ is concave then assume that $M_k \geq X_k$ for some $t = k \leq k_3-1 $, then using the induction step and monotonicity of $\psi^\pi_3$ the result is obtained as in 2.	
	\end{enumerate}
\end{proof}

\begin{remark}
\label{mgrRemark}
Suppose that
\begin{align*}
&\phi^\pi_3(z,k) \leq f(z,k), \forall (z,k) \in \left[ \abs{I_0} ,\theta \right] \times \left\{ 0, \hdots, k_3 \right\} \\
&\pth{\phi^\pi_3(z,k) \geq g(z,k),  \forall (z,k) \in \left[ \abs{I_0} ,\theta \right] \times \left\{ 0, \hdots, k_3 \right\} }
\end{align*}
for some function $f$ ($g$), such that $f$ is concave ($g$ is convex), monotone non-decreasing in $z$, ${\forall (z,k) \in \left[ \abs{I_0} ,\theta \right] \times \left\{ 0, \hdots, k_3 \right\} }$. Then Theorem \ref{MGRThm} can be applied by interchanging $\phi^\pi_3$ by $f$ ($g$), even when $\phi^\pi_3$ does not satisfy the convexity (concavity) condition of the theorem.
\end{remark}

\section{Containment}
\label{containmentSection}

In the previous section we provided a general method for bounding the loss $L^\pi$ when the maximal growth rate is concave or the minimal growth rate convex. Theorem \ref{MGRThm} offers a recursion, which once solved, outputs the bounds explicitly. In many graphs, as we show in Section \ref{explicitBounds}, both $MGR$ and $mgr$ are concave functions. Then, applying Theorem \ref{MGRThm} directly for calculating lower bounds is often not possible. In addition, the recursions defined in Equation (\ref{recursion}) can only be solved numerically. Therefore, in order to obtain explicit theoretical bounds and rules for containment, we bound $MGR$ and $mgr$ in Theorem \ref{containmentThm1} by piecewise linear functions. These allow us to then explicitly solve the recursion relations, calculate explicit bounds, as well as obtain upper and lower bounds for containment. In the second part of this section, we explicitly calculate these bounds on well known topologies including trees, d-dimensional grids, and Erd\H{o}s-R\'{e}nyi graphs.

Let $\Graph{V}{E}$ be a graph on $n$ nodes, and let $s_0 = (I_0, \emptyset)$ be an initial state. Define the function $l_{\alpha,\beta}$ for some parameters $\alpha, \beta \in \mathbb{R}_+$ (which may depend on $p$) as
\begin{equation*}
l_{\alpha,\beta}(p,b,k) = ~
\frac{pb\left(k + 1\right) - \beta}{\alpha} +
\frac{pb}{\alpha^2} 
+ \left[\abs{I_0} - \left(\frac{pb-\beta}{\alpha} + \frac{pb}{\alpha^2} \right) \right]\left(1 + \alpha\right)^{k},
\end{equation*}
and let $k_{b,p}$ be defined by
\begin{equation*}
k_{b,p} = \min 
\left\{
k \in \left\{ 0,1,2,\hdots \right\} :
\frac{\alpha}{p}\left(\alpha\abs{I_0} + \beta\right)
\frac{\left(1+\alpha\right)^{k}}{\left(1+\alpha\right)^{k+1}-1} \leq b
\right\}
\end{equation*}
Here, $l_{\alpha, \beta}$ is the analytical solution to recursion \ref{recursion} when $\phi(z, k)$ is an affine function of the form \mbox{$\alpha z + \beta$}.
The following Theorem is a main result of this paper that lets us calculate explicit containment bounds. The idea behind Theorem \ref{containmentThm1} is as follows. Suppose $MGR$ can be bounded from above by a linear function, and $mgr$ from below by another linear function. Applying Theorem \ref{MGRThm} on these obtains bounds for the loss. These bounds can be written using $l_{\alpha, \beta}$. This in turn enables us to get conditions for containing an infection in expectation. Specifically, we achieve upper and lower bounds for containment dependent on the bounds achieved. For brevity, we use the notation $\tilde{p} = \min \left\{ \Delta p, 1 \right\}$, where $\Delta$ is the maximal degree in $G$.
\begin{theorem}[Containment Bounds]
	\label{containmentThm1}
	Suppose $\exists \alpha,\beta,\gamma,\delta \in \mathbb{R}_+$ for which
	\begin{align}
	\label{GRCondition}
	&\MGR{z,\UC{I_0}{k}}{p} \leq \alpha z + \beta \quad, {\forall (z,k) \in \left[ \abs{I_0} ,\theta \right] \times \left\{ 0, \hdots, k_{b,p} \right\} }, \text{ and} \\
	\label{GRCondition2}
	&\mgr{z,\UC{I_0}{k}}{p} \geq \gamma z + \delta \quad, {\forall (z,k) \in \left[ \abs{I_0} ,\theta \right] \times \left\{ 0, \hdots, k_{b,\tilde{p}} \right\}. }
	\end{align} 
	
	Then
	\begin{equation*}
	b_\theta^U = 
	\begin{cases}
	\min \set{b}{l_{\alpha,\beta}(b, k_{b,p}) \leq \theta} & ,\theta < \infty \\
	\frac{\alpha}{p}\frac{\left(\alpha\abs{I_0} + \beta\right)}{1+\alpha} & ,\theta = \infty,
	\end{cases}
	\end{equation*}
	
	and
	\begin{equation*}
	b_\theta^L = 
	\begin{cases}
	\min \set{b}{l_{\gamma,\delta}(\tilde{p},b, k_{b,\tilde{p}}) \leq \theta} & ,\theta < \infty \\
	\frac{\gamma}{\tilde{p}}\frac{\left(\gamma\abs{I_0} + \delta\right)}{1+\gamma} & ,\theta = \infty
	\end{cases}
	\end{equation*}
	
	are a weak upper and lower bounds for containment, respectively. 
	
	Moreover, 
	\begin{align*}
	&\Loss{\pi}{I_0} \leq l_{\alpha,\beta}(p,b,k_{b,p}) ,\forall b \geq b_\theta^U, \text{ and} \\
	& \Loss{*}{I_0} \geq l_{\gamma,\delta}(\tilde{p},b,k_{b,\tilde{p}}) ,\forall b \geq b_\theta^L,
	\end{align*}
	
	where $\pi$ is some first-order policy.
\end{theorem}

\noindent See appendix for proof.

\begin{corollary}
\label{containmentCorollary}
	Suppose $\exists \alpha,\beta \in \mathbb{R}_+$ for which
\begin{align}
	&\EGR{z,\UC{I_0}{k}}{p} \leq \alpha z + \beta \quad, {\forall (z,k) \in \left[ \abs{I_0} ,\theta \right] \times \left\{ 0, \hdots, k_{b,p} \right\} }, \text{ and} \\
	&\EGR{z,\UC{I_0}{k}}{p} \geq \gamma z + \delta \quad, {\forall (z,k) \in \left[ \abs{I_0} ,\theta \right] \times \left\{ 0, \hdots, k_{b,\tilde{p}} \right\}. }
\end{align}
Then all of the results of Theorem \ref{containmentThm1} hold.
\end{corollary}


\subsection{Obtaining Explicit Bounds}
\label{explicitBounds}
We now turn to finding explicit bounds for well-known topologies, using Theorem \ref{containmentThm1}. To achieve these, we find linear upper bounds for $\MGR{z,\UC{I_0}{k}}{p}$ and linear lower bounds for $\mgr{z,\UC{I_0}{k}}{p}$. Exceptional is the case of regular trees, where the maximal and minimal growth rates coincide under a first-order policy. In that case, the optimal loss can be explicitly calculated, as the upper and lower bounds intersect.

\subsubsection{\textbf{D-Regular Trees}}
\label{fullTreesSubsection}

In Section \ref{TreeExampleSection} we showed a first-order policy,  defined in Algorithm \ref{alg:treePolicy}, is optimal. We give more insight to this topology, and to the reasoning for its first-order optimality.  In Lemma \ref{neighborLemma} we calculate the exact neighborhood cardinality for a set $A$ in a regular tree. We then calculate $L^{\pi^F}$ and show it is in fact equal to $L^*$. Finally, we obtain tight bounds for containment.

\begin{lemma}
	\label{neighborLemma}
	Let $T$ be a regular tree of degree $d$. Let $A$ be a connected set in $T$ such that $root \in A$. 
	
	Then
	\begin{equation*}
	\cut{A} = \abs{N\left(A\right)} = \abs{A}\left(d-1\right)+1.
	\end{equation*}
\end{lemma}

A proof to the lemma can be found in the appendix.

\begin{corollary}
	\label{treeGRCorollary}
	Let $T$ be a regular tree of degree $d$. Suppose $I_0$ is a connected set such that $root \in I_0$. 
	
	Then 
	\begin{equation*}
	\mgr{z,\UC{I_0}{k}}{p} = \MGR{z,\UC{I_0}{k}}{p} = p\left(z\left(d-1\right)+1\right)
	\end{equation*}
\end{corollary}

Corollary \ref{treeGRCorollary} tells us that when $root \in I_0$, the growth rate is equal to the maximal and minimal growth rates. In Theorem \ref{optimalTreePolicyThm} we prove a first-order policy is optimal. Then, applying Theorem \ref{containmentThm1} for the case of $root \in I_0$, we obtain a $tight$ bound for containment. Specifically, letting $\alpha = p(d-1)$ and $\beta = p$, then
\begin{equation*}
b_\theta = 
\begin{cases}
\min \set{b}{l_{p(d-1),p}(b, k_{b,p}) \leq \theta} & ,\theta < \infty \\
p(d-1)\frac{(d-1)\abs{I_0} + 1}{1+p(d-1)} & ,\theta = \infty
\end{cases}
\end{equation*}
is a $tight$ containment bound. \\
Moreover,
\begin{equation*}
\Loss{*}{s_0} = \Loss{\pi^F}{s_0} = l_{\alpha,\beta}(p,b,k_{b,p}) ,\forall b \geq b_\theta.
\end{equation*}	
Note the special case of $I_0 = root$ and $\theta = \infty$ where the $tight$ containment bound becomes  
\begin{equation*} 
b_\infty = \frac{d}{1+\frac{1}{p\left(d-1\right)}}.
\end{equation*}
This result is non-trivial in the sense that as $p$ decreases, a budget lower than $d$ (the degree of the tree) is necessary in order to contain an infection initiated at the root. That is, when $p \neq 1$, even if an infection cannot be stopped at its initial state, it can still be contained at a future state.

\subsubsection{\textbf{D-Dimensional Grid}}

In a conjecture made in \cite{wang2002fire}, at least $2d-1$ firefighters are needed to contain an outbreak starting at a single node on the $d$-dimensional grid. This conjecture was later proven in \cite{wang2010surviving}. In this section we prove a more general result. Specifically, we obtain the minimal number of firefighters (i.e., budget) needed to contain an infection on the $d$-dimensional grid when the infection process is stochastic, and show it agrees with the conjecture for $p = 1$.
 
Let $I_0$ be a set of nodes on the $d$-dimensional grid. Then, it is easy to see that
\begin{equation*}
\MGR{c,\UC{I_0}{k}}{p} \leq 2pdc \quad ,\forall c \in \left\{1,2,\hdots, \abs{V} \right\},
\end{equation*}
where the equality is achieved for a set $I_0$ consisting of $c$ disconnected nodes, each of which has exactly $2d$ outgoing edges.
Next, suppose $I_0$ is known to be a connected set. The maximal cut of a set of cardinality $k$, on an upward crusade initiating at such $I_0$, is given by
$2\left(d - 1 \right)k + 2$. Then by Lemma \ref{GRProperty},
\begin{equation*}
\MGR{z,\UC{I_0}{k}}{p} \leq p\left( 2\left(d - 1 \right)z + 2 \right).
\end{equation*}
Here the set which achieves the maximum is a connected line of cardinality $z$. \\
Applying Theorem \ref{containmentThm1} with $\theta = \infty$ we obtain the weak upper bound
\begin{equation}
\label{eq:gridBound}
b_\infty = \frac{4p\tilde{d}\left(\tilde{d}\abs{I_0} + 1\right)}{1+2p\tilde{d}},
\end{equation}
where $\tilde{d} = d-1$. In the special case of $\abs{I_0} = 1$ and $p = 1$ we then obtain
\begin{equation*}
b_\infty =  \frac{4d(d-1)}{2d-1},
\end{equation*}
which equivalently means an infection can be contained for any $b$ which satisfies
\begin{equation*}
b \geq \ceil{\frac{4d(d-1)}{2d-1}} = \ceil{2d-2 + \frac{2d-2}{2d-1}} = 2d-1.
\end{equation*}
That is, Equation (\ref{eq:gridBound}) proves the bound proven in \cite{wang2010surviving} for a more general case of $p < 1$.

Next, we look for weak lower bounds for containment. In order to use Theorem \ref{containmentThm1}, we must lower bound $mgr$ by a linear function. Unfortunately, on the $d$-dimensional grid, $mgr$ is a concave, non-monotonic function. One method to overcome this problem is to lower-bound $mgr$ by a piecewise linear function, enabling us to apply Theorem \ref{containmentThm1} recursively on finite intervals. Lemma \ref{gridmgr} gives an affine, monotonic non-decreasing lower bound for $mgr$ on the $d$-dimensional grid on a finite interval. The idea of the lemma is to look at balls of radius $r$ around some center node $v$. Next, we define these balls, and a function to calculate their cardinalities, and then use these to obtain an affine lower bound for $mgr$ in Lemma \ref{gridmgr}.

Let $C_{d,r}$ and $N_{c_{d,r}}$  denote the cardinality and neighborhood cardinality of a ball of radius $r$ on the $d$-dimensional grid, respectively.
For clarity, we first consider the 2-dimensional grid. A ball of radius $r$ on the 2-$d$ grid has cardinality 
\begin{equation*}
C_{2,r} = 1 + 2r(r-1), r \in \mathbb{N}.
\end{equation*}
It is not hard to show that the neighborhood cardinality of such a ball is given by $N_{C_{2,r}} = 4r$. \\
Denote by $\left( A_1, A_2, A_3, \hdots \right)$ the sequence of balls of cardinality $C_{2,r}$ and center node $v$, where here $A_1 = \left\{ v \right\}$ and  $\left\{ A_i \right\}_{i = 2}^r$ satisfy
\begin{equation*}
A_i = A_{i-1} \cup N(A_{i-1}).
\end{equation*}
We have that $\abs{A_r} = 1+ \sum_{i=1}^{r-1} N_{C_{2,i}} = C_{2,r}$. It was shown in \cite{wang1977discrete} that these sets minimize the number of neighbors for their respective cardinalities on the 2d grid. Figure \ref{markerGrid} depicts this sequence.

\begin{figure}[t!]
	\centering
	\begin{subfigure}{0.15\textwidth}
		\centering
		\includegraphics[width=0.8\linewidth]{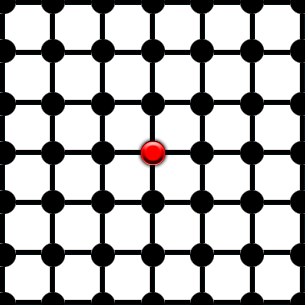}
		\caption{r = 1}
	\end{subfigure}
	\begin{subfigure}{0.15\textwidth}
		\centering
		\includegraphics[width=0.8\linewidth]{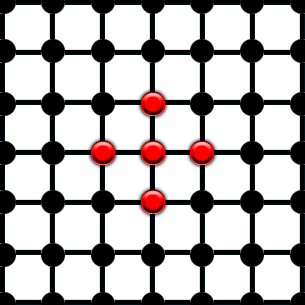}
		\caption{r = 2}
	\end{subfigure}
	\begin{subfigure}{0.15\textwidth}
		\centering
		\includegraphics[width=0.8\linewidth]{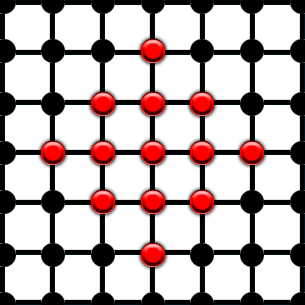}
		\caption{r = 3}
	\end{subfigure}
	\caption{ Plotting configurations of minimal number of neighbors on the 2d-grid for sets of cardinality\\ $C_{2,r} = 1 + 2r(r-1)$ }
	\label{markerGrid}
\end{figure}

\noindent The case of $d$-dimensional balls is similar. We define the function $\nu_d(r)$ recursively by
\begin{align*}
& \nu_2(r) = 4r \\
& \nu_d(r) = \nu_{d-1}(r) + 2\left( 1 + \sum_{i = 0}^{r-1} \nu_{d-1}(i) \right) \quad ,d \geq 3.
\end{align*}
Here $\nu_d(r)$ are cardinalities of neighborhoods for balls of radius $r$ on the $d$-dimensional grid.  
Then \begin{equation*}
C_{d,r} = 1 + \sum_{i = 0}^{r-1} \nu_{d}(i).
\end{equation*}
Having defined $C_{d,r}$ and $\nu_d(r)$ we can now lower bound $mgr$ by a piecewise linear function. 

\begin{lemma}
	\label{gridmgr}
	Let $\Graph{V}{E}$ be a d-dimensional grid, and let $a,\theta \in \mathbb{N}$ with $a < \theta$. Then
	\begin{align*}
	&\mgr{c,\UC{I_0}{k}}{p} \geq p \left( m_{a,\theta}(c-a) + \nu_d(r_a)\right) \quad ,c \in \left\{ a, \hdots, \theta \right\},
	\end{align*}
	where
	\begin{equation*}
	m_{a,\theta} = \frac{\nu_d(r_\theta) - \nu_d(r_a)}{\sum_{i = r_a-1}^{r_\theta-1} \nu_{d}(i)} 
	\quad
	\text{ and }
	\quad 
	r_z = \min \set{r}{ z \leq 1 + \sum_{i = 0}^{r-1} \nu_{d}(i)}.
	\end{equation*}
\end{lemma}
See Appendix for proof.
	
	\subsubsection{\textbf{Erd\H{o}s-R\'{e}nyi}}
	
	Let $G_{n,s}$ be an Erd\H{o}s-R\'{e}nyi model, where $n$ denotes the number of nodes, and $s$ the edge sampling probability.
	Then for any set $A$, we have that 
	\mbox{$\E\cut{A} = s\abs{A}\left(n - \abs{A}\right)$} and
	\mbox{$\E\abs{N(A)} = \left(n - \abs{A}\right)\left(1 - \left(1-s\right)^{\abs{A}}\right)$}.
	
	\noindent Then, by Lemma \ref{GRProperty}, we can bound the expected growth rate by
	\begin{align}
	&\EGR{z,\UC{I}{k}}{p} \leq spz\left(n - \abs{I}\right), \text{ and}
	\label{eq:erdosBound1} \\
	&\EGR{z,\UC{I}{k}}{p} \geq \left(1 - \left(1-s\right)^{z}\right)p\left(n - z\right).
	\label{eq:erdosBound2}
	\end{align}

	We consider the sparse regime, where $s = \frac{c}{n}$ for some constant $c > 1$. We can upper bound the expected growth rate in equation (\ref{eq:erdosBound1}) further by a linear function
	$
	\EGR{z,\UC{I}{k}}{p} \leq spzn = cpz
	$.
	Then, for $\theta = \infty$ a weak upper bound for containment is given by
	\begin{equation*}
	b_\infty = \frac{c^2p}{1+cp}\abs{I_0}.
	\end{equation*}
	When $cp \gg 1$ the bound becomes $b_\infty = c\abs{I_0}$. \\
	When $cp \ll 1$ the bound becomes $b_\infty = c^2p \abs{I_0}$.
	
	Similarly to the $d$-dimensional grid, the lower bound of the growth rate on the Erd\H{o}s-R\'{e}nyi graph is concave. This in turn means that in order to apply Theorems \ref{MGRThm} and \ref{containmentThm1} we must approximate a linear lower bound function on finite intervals. Using equation (\ref{eq:erdosBound2}) we can obtain an affine, monotonic non-decreasing lower bound. More specifically,
	\begin{equation*}
	\EGR{z,\UC{I}{k}}{p} \geq \frac{g_n(\theta_n^{max})}{\theta_n^{max}}pz, \forall z \in \left\{1, \hdots, \theta_n^{max} \right\},
	\end{equation*}
	where $g_n: \mathbb{N} \to \mathbb{R_+}$ is defined by
	$
	g_n(z) = \left(n - z\right)\left(1 - \left(1-\frac{c}{n}\right)^{z}\right)
	$ and 
	$
	\theta_n^{max} = \argmax_{z \in \mathbb{N}} g_n(z)
	$

\subsection{State-Dependent Budget}
\label{stateDependentBudgetSection}

In the previous section we have shown how Theorem \ref{containmentThm1} can be applied to well known toplogies that have been studied extensively in the literature. In this section we show how Theorem \ref{containmentThm1} can be used to construct a state-dependent budget allocation policy, which suggests a specific budget for every time step, according to the state's current $EGR$/$mgr$ characteristic. We show this strategy acheives better containment on two real world networks, when compared to a constant budget strategy which consumes an equal global budget.

Recall Theorem \ref{containmentThm1} gives guarantees for containtment of an infection when $mgr$ can be lower bounded by a linear function under a minimal budget $b^L$. In fact, when only a lower bound in $\left[\abs{I_0}, \theta \right]]$ is given, then $b^L = b^L_\theta$ is a lower bound for containment in $\theta$ (i.e., the infection can be contained in expectation with cardinality at most $\theta$). In theory, given $mgr/EGR$, a controller can choose at each state $s_t = \pth{I_t,B_t}$ an objective $\theta_t = \abs{I_t} + M_t$, for some $M_t \in \mathbb{N}$. Then, by lower bounding $mgr/EGR$ in $\left[ \abs{I_t}, \theta_t \right]$ and applying Theorem \ref{containmentThm1}, a specific budget can be allocated for that state. This process can be repeated at each time step in order to obtain a ``minimal" budget according to the current lower bound in $\left[ \abs{I_t}, \theta_t \right]$. 

In contrast to trees and grids (as well as other symmetric topologies), an explicit expression of $mgr/EGR$ is hard to obtain. Even a lower bound may be complicated to express. For this reason, Monte Carlo simulations can be used in order to approximate $mgr/EGR$. Specifically, let $\tau = \braces{u_1, u_2, \hdots, u_T}$ be $T$ trajectories of length $d$ starting at state $s$. We can approximate a lower bound on $mgr$ for $d$ steps by
\begin{equation}
\widehat{mgr}(k) \geq p \minl_{u \in \tau} \minl_{s_i \in u : \abs{s_i} = k} \abs{N(s_i)} \quad, k \in \braces{0, \hdots d}.
\label{eq:lb1}
\end{equation}
Similarly, a lower bound on $EGR$ can be approximated by
\begin{align}
&n_k = \sum_{u \in \tau}\sum_{s_i \in u}\indicator{\abs{s_i} = k} \quad k \in \braces{0, \hdots d} \nonumber \\
&\widehat{EGR}(k) \geq \frac{p}{n_k} \sum_{u \in \tau}\sum_{s_i \in u}N(s_i)\indicator{\abs{s_i} = k}.
\label{eq:lb2} 
\end{align}
Using the lower bounds in Equations (\ref{eq:lb1}) and (\ref{eq:lb2}) a simple linear lower bound can be constructed. Finally, Theorem \ref{containmentThm1} can be applied to these linear lower bounds.  Algorithm \ref{alg:dependentBudget} gives a state-dependent budget allocation procedure which applies Theorem \ref{containmentThm1} with an approximated linear lower bound to $mgr/EGR$ using Monte Carlo simulation.

\subsubsection{Experiments}
We have run our simulations on two real email networks: the Enron email communication network \cite{klimt2004introducing}, which covers all the email communication within a dataset of approximately half million emails, and the EU email communication network \cite{leskovec2007graph}, generated using email data from a large European research institution for a period of 18 months.
In part, using email networks is motivated by the fact that many computer viruses spread by email attachments. In our simulations we have used the first-order CUT policy, as defined in Definition \ref{def:CUT}. More specifically, we have chosen to test our algorithm on these networks due to their broad degree distributions. first-order policies are highly efficient on such topologies, allowing an unbiased assessment of our budget allocation algorithm.

\begin{figure}[t!]
\centering
\begin{subfigure}{0.49\textwidth}
	\centering
	\includegraphics[width=0.95\linewidth]{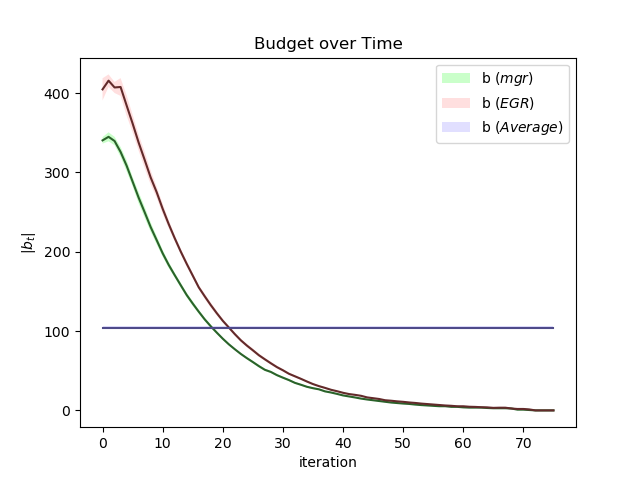}
	\caption{}
\end{subfigure}
\begin{subfigure}{0.49\textwidth}
	\centering
	\includegraphics[width=0.95\linewidth]{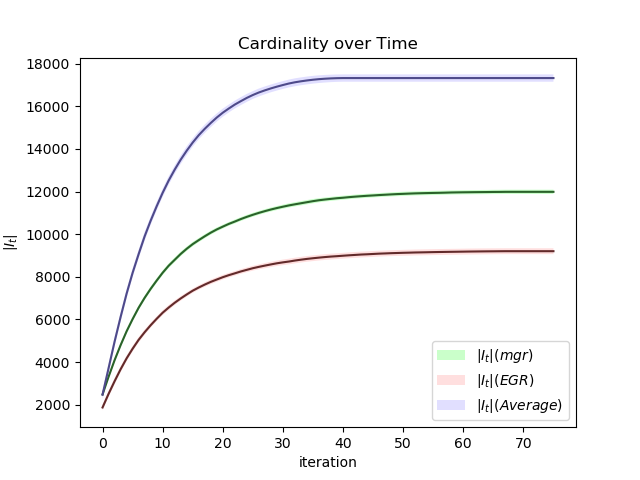}
	\caption{}
\end{subfigure}
\begin{subfigure}{0.49\textwidth}
	\centering
	\includegraphics[width=0.95\linewidth]{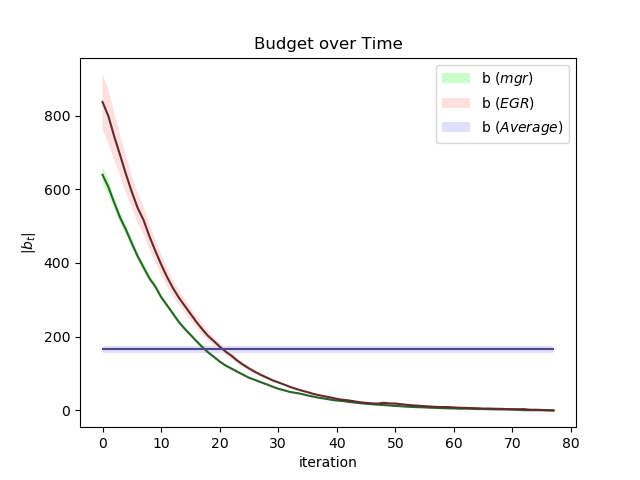}
	\caption{}
\end{subfigure}
\begin{subfigure}{0.49\textwidth}
	\centering
	\includegraphics[width=0.95\linewidth]{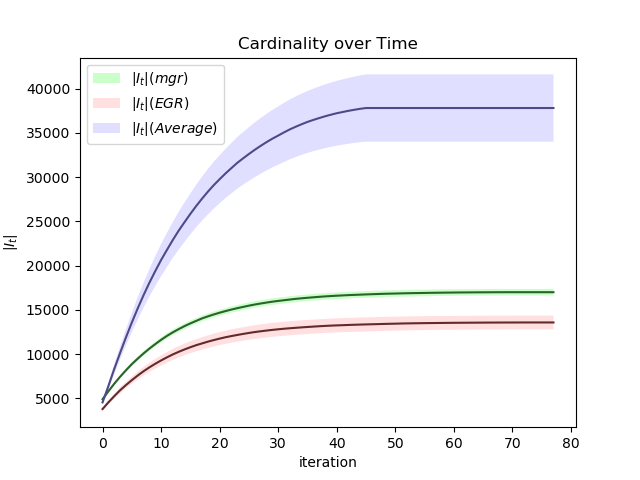}
	\caption{}
\end{subfigure}
\caption{ Plots comparing different budget allocation strategies on the Enron email network of 36,692 nodes (Plots a,b) and EU research institution email network of 265,214 nodes (Plots c,d). Plots show state dependent budget allocation using $mgr$ (green), $EGR$ (red), and constant budget equal to the average budget used by $EGR$ (blue). Initial cardinalities that were tested are 2000 infected nodes for Enron network, and 5000 infected nodes for EU network. Infection rate was $p = 0.05$ for both networks. Both networks were sampled 100 times and averages over 30 runs per sample. Light areas show standard deviation from mean. }
\label{fig:budgetAllocation}
\end{figure}

Results of our simulations are plotted in Figure \ref{fig:budgetAllocation}. Statistical data on the two networks in provided in Table \ref{table:stats}. In all of our experiments we picked starting nodes uniformly at random \footnote{Choosing sources in a realistic way is an open problem - the data that could offer a solution to this problem seems to be extremely scarce at this time.}. We averaged the results over 100 uniform samples of initial states. We used initial cardinalities of 2000 for the Enron network and 5000 for the EU network. Each sample was run 30 times. We've tested lower bounds on $EGR$ and $mgr$ using Monte Carlo simulations of length 3. All our experiments used an infection probability of $p = 0.05$. 

We compared our results to a constant budget allocation strategy. The constant budget was calculated by taking the maximal total budget used by our algorithm and dividing by the average number of iterations used. Specifically, denote by $b^{(i)}_{EGR}(t)$ and $b^{(i)}_{mgr}(t)$ the budget used in test $i$ by the $EGR$ and $mgr$ lower approximations, respectively. Also, denote by $T_{avg}$ the average number of iterations for containment under these approximations. Then,
\begin{equation*}
b_{global} = \frac{\maxl_i \sum_{t=0}^\infty \max\braces{b^{(i)}_{EGR}(t),b^{(i)}_{mgr}(t)} }{T_{avg}}.
\end{equation*}
Results, as depicted in Figure \ref{fig:budgetAllocation}, clearly show that a state dependent budget strategy which allocates more resources at the beginning of an infection is superior over one that only uses a constant budget with an equal global budget.

\begin{algorithm}[t]
	\SetAlgoNoLine
	\KwIn{state $s = (I,B)$, infection rate $p$, lower bound type $type$.}
	\KwOut{Budget $b(s)$}
	$\tau \gets$ Sample $T$ trajectories of $d$ iterations each starting at state $s$ with zero budget.\;
	$\tilde{\Delta} \gets$ MaxDegree($\tau$)\;
	$\tilde{p} \gets \min \{1, \tilde{\Delta}p\}$\;
\For{$k \in \curly{\abs{I}, \hdots \abs{V}}$}{
~~	\uIf{$type = $"$mgr$"}{
~~~~    $\hat{LB}(k,\tau) \gets p \minl_{u \in \tau} \minl_{s_i \in u : \abs{s_i} = k} \abs{N(s_i)}$ \;
  }
~~  \uElseIf{$type = $"$EGR$"}{
~~~~  $n_k \gets \sum_{u \in \tau}\sum_{s_i \in u}\indicator{\abs{s_i} = k}$ \;
~~~~   $\hat{LB}(k,\tau) \gets \frac{p}{n_k} \sum_{u \in \tau}\sum_{s_i \in u}N(s_i)\indicator{\abs{s_i} = k}$ \;
  }
 }
	$\theta \gets \argmax_{k > \abs{I}} \hat{LB}(k,\tau)$\;
	$\alpha \gets \frac{\hat{LB}(\theta,\tau) - \hat{LB}(\abs{I},\tau)}{\theta - \abs{I}}$\;
	$\beta \gets \hat{LB}(\abs{I},\tau)$\;
	\emph{Solve for b:}
	$\min{b} : l_{\alpha,\beta}(\tilde{p},b,k_{b,\tilde{p}}) \leq \theta$\;
	\Return $b$
\caption{State Dependent Budget Allocation}
\label{alg:dependentBudget}
\end{algorithm}

\begin{table}[]
\centering
\begin{tabular}{|c|c|c|c|c|c|c|c|}
\hline
      & Nodes  & Edges  & \begin{tabular}[c]{@{}c@{}}Nodes in \\ largest\\ WWC\end{tabular} & \begin{tabular}[c]{@{}c@{}}Nodes in\\ largest\\ SSC\end{tabular} & \begin{tabular}[c]{@{}c@{}}Average\\ clustering\\ coefficient\end{tabular} & \begin{tabular}[c]{@{}c@{}}Number of\\ triangles\end{tabular} & Diameter \\ \hline
Enron & 36692  & 183831 & 33696                                                             & 33696                                                            & 0.4970                                                                     & 727044                                                        & 11       \\ \hline
EU    & 265214 & 420045 & 224832                                                            & 34203                                                            & 0.0671                                                                     & 267313                                                        & 14       \\ \hline
\end{tabular}
\caption{Dataset statistics of the Enron and EU email networks.}
\label{table:stats}
\end{table}

\section{Conclusion}
In this paper we proposed several approaches for immunization and containment of "fast-moving" (i.e., not infinitesimally slow) epidemics in networks. We modeled the problem of epidemic spread as a stochastic extension of the Firefighter problem. 

Our work focused on the question when containment can be guaranteed in expectation? And conversely, when is the spread of an epidemic probable?
We found it easier to obtain weak upper bounds for containment than lower bounds. This is due to the fact that the MGR is more commonly concave, allowing the use of Theorem 5.6 for upper bounding the loss. We  showed that for some topologies, such as d-dimension grids and Erd\H{o}s R\'{e}nyi graphs, mgr is concave, not enabling us to apply Theorem 5.6 directly for lower bounding the optimal loss. We therefore approximated mgr by lower bounding it by a piecewise linear function, then applied Theorem \ref{containmentThm1} to obtain an approximate weak lower bound for containment. 

Our bounds form gaps between the weak upper and lower bounds for containment. An interesting question, not addressed in this paper is thus, can we make this gap small? We showed that there is no gap for regular trees, and that our bounds are tight. It is interesting to analyze then, how loose is this gap on different topologies? And more interestingly, what is the right dependence of the budget $b$, the infection probability $p$, and the topology, so that a budget larger than $b$ or an infection probability smaller than $p$ lead to containment, and vice versa: a budget smaller than $b$ or infection probability higher than $p$ leads to unbounded spreading of the infection. 

In Section \ref{stateDependentBudgetSection} we've constructed an algorithm that uses approximated lower bounds on $mgr/EGR$ for state dependent budget allocation. We've tested this algorithm using a first order policy and showed such a budget allocation strategy outperforms constant budget allocation with an equal global budget.

Finally, it is important to develop robust networks, where containment of infections can be established quickly with minimal casualties. Our containment bounds can be used as criteria for network robustness. Minimizing the budget needed for containment may create networks that are easier to manage under infectious attacks. Our work can be considered a first step towards the design of networks that have built-in resilience by design.

\bibliography{bibfile}
\bibliographystyle{iclr2018_conference}

\section{Appendix (Missing Proofs)}

\subsection*{Proof of Theorem \ref{containmentThm1}}
To prove the Theorem we require a few more preliminary extensions to our definitions of growth rates, as well as rules connecting the different types of bounds on growth rates.

\subsubsection*{Definition Extensions}
We begin by extending some of the definitions related to growth rates. Similar to the definitions of (maximal/minimal) growth rates, denoted by $MGR,mgr$ (see Section \ref{growthRateSection}), we define (maximal/minimal) stop rates, denoted by $MSR,msr$, to relate to the growth rate that would have been into vaccinated nodes under a given upward crusade.

 \begin{definition} [Maximal and Minimal Stop Rates]
	Given a policy $\pi \in \Pi$, a state $s$, and the set of upward crusades of length $k$, $\UC{\pi,s}{k}$, the maximal and minimal stop rates are functions 	${MSR, msr: \left\{ 1, \hdots, \abs{V} \right\} \times \UC{\Pi,s}{k} \to \mathbb{R}_+}$, 	defined by
	\begin{align*}
	&\MSR{c,\UC{\pi,s}{k}}{p} = \maxl_{\substack{u \in \UC{\pi,s}{k} \\ u_k \in S_{c} }} \GR{I_k,B_k}{p} , \text{ and} \\
	&\msr{c,\UC{\pi,s}{k}}{p} = \minl_{\substack{u \in \UC{\pi,s}{k} \\ u_k \in S_{c} }} \GR{I_k,B_k}{p} ,
	\end{align*}
	
	where $S_{c} = \set{s = (I,B)}{ \abs{I} = c}$, and recall that $u_k = (I_k,B_k)$ denotes the final state in an upward crusade $u$. 
	\end{definition}
\noindent We also consider the expected stop rate, as defined below.
\begin{definition} [Expected Stop Rate]
	Given a policy $\pi \in \Pi$, a state $s$, and the set 	of upward crusades of length $k$, $\UC{\pi,s}{k}$, the expected stop rate is a function	
	${ESR: \UC{\Pi,s}{k} \to \mathbb{R}_+}$, 
	defined by
\begin{equation*}
	\ESR{c, \UC{\pi,s}{k}}{p} = \E^\pi \giventhat{  \GR{I_k,B_k}{p}} { \abs{I_k} = c, s_0 = s}.
	\end{equation*}
	\noindent \mbox{where $u_k = (I_k, B_k)$ is the random variable of the final state in upward crusades of $\UC{\pi,s}{k}$}.
	\end{definition}
	It is sometimes useful to work with empty policies (i.e., policies which do not vaccinate any nodes). For the case of (maximal/minimal) growth rates we use empty upward crusades, denoted by $\UC{I}{k}$. Contrary to these, (maximal/minimal) stop rates require a different definition.
	\begin{definition} [Maximal and Minimal Stop Rates - empty policies]
	Given a set of infected nodes $I$, and empty upward crusades of length $k$, $\UC{I}{k}$, 	the maximal and minimal stop rates are functions ${MSR, msr: \left\{ 1, \hdots, \abs{V} \right\} \times \left\{ 1, \hdots, \abs{V} \right\} \times \UC{V}{k} \to \mathbb{R}_+}$, 
	defined by
	\begin{align*}
	&\MSR{c_1,c_2,\UC{I}{k}}{p} = \maxl_{\substack{u \in \UC{I}{k} \\ u_k \in S_{c_1} }} \maxl_{\abs{B_k} = c_2}\GR{I_k,B_k}{p}, \text{ and} \\
	&\msr{c_1,c_2,\UC{I}{k}}{p} = \minl_{\substack{u \in \UC{I}{k} \\ u_k \in S_{c_1} }}  \minl_{\abs{B_k} = c_2}\GR{I_k,B_k}{p},
	\end{align*}
	
	where $S_{c} = \set{s = (I,\emptyset)}{ \abs{I} = c }$, and recall that $u_k = (I_k,\emptyset)$ denotes the final state 	in an upward crusade $u$. 
	\end{definition}

The following lemma forms a connection between the definitions of (maximal/minimal/expected) growth rates and (maximal/minimal) stop rates.

\begin{lemma}
\label{GRProperties}
Let $\pi \in \Pi$ and $s = (I,B)$, then
\begin{enumerate}
  \item $\abs{N(I)\cap B}p \leq \GR{I,B}{p} \leq \min \left\{ \cut{I,B}p,1 \right\}$.
	\item $\MGR{z,\UC{\pi,s}{k}}{p} \leq \subplus{\MGR{z,\UC{I}{k}}{p} - \msr{z,\UC{\pi,s}{k}}{p}}$	.
	\item $\mgr{z,\UC{\pi,s}{k}}{p} \geq \subplus{\mgr{z,\UC{I}{k}}{p} - \MSR{z,\UC{\pi,s}{k}}{p}}$.
	\item $\EGR{z,\UC{\pi,s}{k}}{p} = \subplus{\EGR{z,\UC{I}{k}}{p} - \ESR{z,\UC{\pi,s}{k}}{p}}$.
	\item ${\MSR{z,\UC{\pi,s}{k}}{p} \leq \MSR{z,b(k+1) + \abs{B},\UC{I}{k}}{p} \leq \min \left\{ p \Delta,1 \right\}\left( b(k+1)+ \abs{B}\right)}$.
	\item $\msr{z,\UC{\pi,s}{k}}{p} \geq \msr{z,b(k+1),\UC{I}{k}}{p}$.
	\item If in addition $\pi$ is a first order policy and $B = \emptyset$ then \\ 
	$\msr{z,\UC{\pi,s}{k}}{p} \geq \msr{z,b(k+1),\UC{I}{k}}{p} \geq pb(k+1)$.
\end{enumerate}
\end{lemma}
	\begin{proof}
		\noindent1) Using the same steps as in the proof of Lemma \ref{GRProperty}, for any $p \in \left[0, 1\right]$ and $k \in \mathbb{N}$,
		\begin{equation*}
		\GR{s}{p} = \sum_{v \in N(s)} \left(1-(1-p)^{\cut{s,v}}\right) \leq  \sum_{v \in N(s)} \min \left\{ p\cut{s,v}, 1 \right\} = \min \left\{ \cut{s} p, 1 \right\}.
		\end{equation*}
		On the other hand, for any $v \in N(s)$,
		\begin{equation*}
		\cut{s,v} \geq 1.
		\end{equation*}
		Then for any $p \in \left[0, 1\right]$,
		\begin{equation*}
		\GR{s}{p} = \sum_{v \in N(s)} \left(1-(1-p)^{\cut{s,v}}\right) \geq \sum_{v \in N(s)} \left(1-(1-p)\right) = \abs{N(s)} p.
		\end{equation*}
		
		\noindent2) By equation (\ref{GRSeparationEq}), for a given state $s = (I,B)$, the growth rate can be written as
		\begin{equation*}
		\GR{s}{p} = \GR{I}{p} - \GR{I,B}{p}.
		\end{equation*}
		Then,
		\begin{align*}
		&\MGR{c,\UC{\pi,s}{k}}{p} = 
		\maxl_{\substack{u \in \UC{\pi,s}{k} \\ u_k \in S_{c} }} \GR{u_k}{p} =
		\maxl_{\substack{u \in \UC{\pi,s}{k} \\ u_k \in S_{c} }} \left\{ \GR{I_k}{p} - \GR{I_k,B_k}{p} \right\} \leq \\
		&\leq \maxl_{\substack{u \in \UC{\pi,s}{k} \\ u_k \in S_{c} }} \GR{I_k}{p}
		- \minl_{\substack{u \in \UC{\pi,s}{k} \\ u_k \in S_{c} }} \GR{I_k,B_k}{p}\leq \\
		& \leq \maxl_{\substack{u \in \UC{I}{k} \\ u_k \in S_{c} }} \GR{I_k}{p}
		- \minl_{\substack{u \in \UC{\pi,s}{k} \\ u_k \in S_{c} }} \GR{I_k,B_k}{p} =
		\MGR{z,\UC{I}{k}}{p} - \msr{z,\UC{\pi,s}{k}}{p},
		\end{align*}
		where the last inequality follow by Remark \ref{UCInclusionRemark}. \\
		Finally since the growth rate is non-negative, so is the maximal growth rate. Hence
		\begin{equation*}
		\MGR{z,\UC{\pi,s}{k}}{p} \leq \subplus{\MGR{z,\UC{I}{k}}{p} - \msr{z,\UC{\pi,s}{k}}{p}}.
		\end{equation*}
		
		\noindent3) The proof follows the same steps as in (2).
		
		\noindent4) Similar to (2) we have that
		\begin{equation*}
		\GR{s}{p} = \GR{I}{p} - \GR{I,B}{p}.
		\end{equation*}
		Then,
		\begin{align*}
		&\E^\pi \giventhat{ \GR{s_k}{p} } { \abs{I_k} = c, s_0 = s} = \\
		&\E^\pi \giventhat{  \GR{I_k}{p} } { \abs{I_k} = c, s_0 = s} - \E^\pi \giventhat{  \GR{I_k,B_k}{p}} { \abs{I_k} = c, s_0 = s}= \\
		&=	\EGR{z,\UC{I}{k}}{p} - \ESR{z,\UC{\pi,s}{k}}{p},
		\end{align*}
		Finally since the growth rate is non-negative, so is the expected growth rate. Hence
		\begin{equation*}
		\EGR{z,\UC{\pi,s}{k}}{p} = \subplus{\EGR{z,\UC{I}{k}}{p} - \ESR{z,\UC{\pi,s}{k}}{p}}.
		\end{equation*}

		\noindent5) Since we only consider policies which vaccinate exactly $b$ nodes at every iteration, after $k$ iterations of reaching state $s = (I,B)$ exactly $\abs{B} + b(k+1)$ nodes will be vaccinated. In addition, any upward crusade in $\UC{\pi,s}{k}$ exists as a sequence of sets for an upward crusade in $\UC{I}{k}$. Hence,
		\begin{align*}
		& \MSR{c,\UC{\pi,s}{k}}{p} = 
		\maxl_{\substack{u \in \UC{\pi,s}{k} \\ u_k \in S_{c} }} \GR{I_k,B_k}{p}
		= \maxl_{\substack{u \in \UC{\pi,s}{k} \\ u_k \in S_{c} }} \maxl_{\abs{B_k} = \abs{B} + b(k+1)} \GR{I_k,B_k}{p} \leq \\
		& \leq \maxl_{\substack{u \in \UC{I}{k} \\ u_k \in S_{c} }} \maxl_{\abs{B_k} = \abs{B} + b(k+1)} \GR{I_k,B_k}{p}  = \MSR{c,\abs{B} + b(k+1),\UC{I}{k}}{p}.
		\end{align*} 
		
		Next, using part (1), $\GR{I,B}{p} \leq \cut{I,B}p$. Then for any time $k \geq 0$,
		\begin{align*}
		&\GR{I_k,B_k}{p} \leq \cut{I_k,B_k}p = \sum_{v \in B_k} \min \left\{ p \cut{I_k,v},1 \right\} \leq \sum_{v \in B_k} \min \left\{ p \Delta,1 \right\} = \\
		&= \min \left\{ p \Delta , 1 \right\} \abs{B_k} =  \min \left\{ p \Delta, 1 \right\} \left( b(k+1)+ \abs{B}\right).
		\end{align*}

		6) The proof follows the same steps as in (4).
		
		7) By (5) we know that $\msr{z,\UC{\pi,s_0}{k}}{p} \geq \msr{z,b(k+1),\UC{I_0}{k}}{p}$. \\
	 Next, since $\pi$ is a first order policy, $B_k \subseteq N(I_k)$ for all $k \geq 0$. Then using (1),
	 \begin{equation*}
	 \GR{I_k,B_k}{p} \geq p\abs{N(I_k)\cap B_k} = p\abs{B_k} = pb(k+1).
	 \end{equation*}
	\end{proof}

\subsubsection*{Proof of Theorem}
We continue with the proof of the Theorem, and choose to start by proving the upper bounds. \\
Let $\pi$ be a first order policy. The proof will be as follows. We begin by applying Theorem \ref{MGRThm} with results obtained in Lemma \ref{GRProperties} in order to find an explicit recursion relation $X_k = f(X_{k-1})$ which bounds the process $M_k^\pi = \E^{\pi} \giventhat{\abs{I_k}}{s_0}$ from above. We then solve the recursion relation and find conditions under which the condition $X_k \leq \theta$ is satisfied. When the latter is satisfied for all $k \geq 0$, the infection is contained, and thus we obtain a weak lower bound for containment in $\theta$. We can then also calculate the exact upper bound, which in turn gives us an upper bound for $\Loss{\pi}{s_0}$.
		
		By Lemma \ref{GRProperties} (2),
		\begin{equation*}
		\MGR{z,\UC{\pi,s_0}{k}}{p} \leq \subplus{\MGR{z,\UC{I_0}{k}}{p} - \msr{z,\UC{\pi,s_0}{k}}{p}}.
		\end{equation*}
		By Lemma \ref{GRProperties} (6) we also know that for $s_0 = (I_0,\emptyset)$ and a first order policy $\pi$,
		\begin{equation*}
		\msr{z,\UC{\pi,s_0}{k}}{p} \geq pb(k+1).
		\end{equation*}
		Then, using the above together with equation \ref{GRCondition} we have
		\begin{equation*}
		\MGR{z,\UC{\pi,s_0}{k}}{p} \leq \subplus{\alpha_p z + \beta_p - pb(k+1)} \quad, {\forall (z,k) \in \left[ \abs{I_0} ,\theta \right] \times \left\{ 0, \hdots, k_{b,p} \right\} }.
		\end{equation*}
		The function $f(z) = \subplus{\alpha_p z + \beta_p - pb(k+1)}$ is concave and monotonic non-decreasing in $z$ for all $z \geq 0$. Thus, by Remark \ref{mgrRemark} of Theorem \ref{MGRThm}, the process $M_t^\pi = \E^{\pi} \giventhat{\abs{I_t}}{s_0}$ is upper bounded by the solution to the recursion relation
		\begin{align*}
		&X_{k+1} = X_{k} + \subplus{ \alpha_p X_{k} + \beta_p - pb(k+1) }  &,k \geq 0, ~X_k \leq \theta \\
		&X_0 = \abs{I_0}
		\end{align*}
		which can be rewritten
		\begin{align*}
		&X_{k} = X_{k-1} + \subplus{  \alpha_p X_{k-1} + \beta_p - pbk }  &,k \geq 1, ~X_k \leq \theta \\
		&X_0 = \abs{I_0}
		\end{align*}
		and can be further rewritten
		\begin{align}
		&X_{k} = \left(1+\alpha_p\right) X_{k-1} + \beta_p - pbk  &,k \geq 1, ~X_k \leq \theta, ~\alpha_p X_{k-1} + \beta_p - pbk \geq 0  \nonumber \\
		&X_{k} = X_{k-1}  &,k \geq 1, ~X_k \leq \theta,  ~\alpha_p X_{k-1} + \beta_p - pbk < 0  \nonumber \\
		&X_0 = \abs{I_0} \label{contrained}
		\end{align}
		We also consider an unconstrained formulation of the recursion
		\begin{align}
		&\tilde{X}_{k} = \left(1+\alpha_p\right) \tilde{X}_{k-1} + \beta_p - pbk  &,k \geq 1, ~\alpha_p \tilde{X}_{k-1} + \beta_p - pbk \geq 0  \nonumber \\
		&\tilde{X}_{k} = \tilde{X}_{k-1}  &,k \geq 1,  ~\alpha_p \tilde{X}_{k-1} + \beta_p - pbk < 0  \nonumber \\
		&\tilde{X}_0 = \abs{I_0} \label{uncontrained}
		\end{align}
		We bring the following simple auxiliary lemma
		\begin{AuxLemma}
			\label{recursionLemma}
			The solution to the recursion relation
			\begin{flalign*}
			&X_k = aX_{k-1} + bk + c  ,k \geq 1\\
			&X_0 = m
			\end{flalign*}
			is given by
			\begin{equation*}
			X_k = \frac{c+bk}{1-a} - \frac{ab}{\left(1-a\right)^2} + 
			\left(m - \left(\frac{c}{1-a} - \frac{ab}{\left(1-a\right)^2}\right)\right)a^k.
			\end{equation*}
		\end{AuxLemma} 
		\noindent If we solve recursion \ref{uncontrained} and obtain a result $\tilde{X}_k$ such that $\tilde{X}_k \leq \theta$ for all $k \geq 1$, then obviously $\tilde{X}_k = X_k$ for all $k \geq 1$, where $X_k$ is the solution to recursion \ref{contrained}. This is true since the condition $X_k \leq \theta$ is never active. Suppose $k \geq 1$ and satisfies $\alpha_p X_{k-1} + \beta_p - pbk \geq 0$. Also suppose $X_k, \tilde{X_k} \leq \theta$, such that $X_k = \tilde{X_k}$. Solving recursion \ref{contrained} using the auxiliary lemma yields
		\begin{equation}
		\label{recursionSol}
		X_k = \tilde{X}_k = 
		\frac{pb\left(k + 1\right) - \beta_p}{\alpha_p} +
		\frac{pb}{\alpha_p^2} 
		+ \left[\abs{I_0} - \left(\frac{pb-\beta_p}{\alpha_p} + \frac{pb}{\alpha_p^2} \right) \right]\left(1 + \alpha_p\right)^{k} = l_{\alpha_p,\beta_p}(p,b,k).
		\end{equation}
		
		Let us first consider the case of $\theta < \infty$.
		Denote by $k_{b,p}$ the maximal time under which $\alpha_p X_{k-1} + \beta_p - pbk \geq 0$ (or equivalently, the minimal time under which $\alpha_p X_{k-1} + \beta_p - pbk \leq 0$). If such $k_{b,p}$ exists then it also holds that ${X_k = \tilde{X}_k = l_{\alpha_p,\beta_p}(b, k_{b,p})}$ for all $k \geq k_{b,p}$. 
		Next take $b$ which satisfies ${l_{\alpha_p,\beta_p}(b, k_{b,p}) \leq \theta}$. The process $M_k^\pi = \E^{\pi} \giventhat{\abs{I_k}}{s_0}$ is then bounded from above by the solution to $X_k$ for any $k \geq 0$. Specifically since by definition of $k_{b,p}$, $X_k$ does not change for $k \geq k_{b,p}$, we have that $X_{k_{b,p}} = X_\infty \geq M_\infty^\pi$, hence
		\begin{equation*}
		\Loss{\pi}{I_0} = M_{T^*}^\pi = M_\infty^\pi \leq X_\infty = X_{k_{b,p}} = l_{\alpha_p,\beta_p}(p,b,k_{b,p}).
		\end{equation*}
		This is true for any $b$ which satisfies ${l_{\alpha_p,\beta_p}(b, k_{b,p}) \leq \theta}$, specifically for
		\begin{equation*}
		b_\theta = \min \set{b}{l_{\alpha_p,\beta_p}(b, k_{b,p}) \leq \theta}.
		\end{equation*}
		This implies, by definition, that $b_\theta$ is a weak upper bound for containment in $\theta$.
		
		Next let us consider the case of $\theta = \infty$.
		Substituting $X_{k-1}$ in $\alpha_p X_{k-1} + \beta_p - pbk \leq 0$ by the solution obtained in Equation \ref{recursionSol} yields
		\begin{equation*}
		\alpha_p l_{\alpha_p,\beta_p}(p,b,k-1) + \beta_p - pbk \leq 0.
		\end{equation*}
		Substituting the expression for $l_{\alpha_p,\beta_p}(p,b,k-1)$ and solving for $b$ then yields
		\begin{equation}
		\label{weakUpperBoundEq}
		b \geq \frac{\alpha_p}{p}\left(\alpha_p\abs{I_0} + \beta_p\right)
		\frac{\left(1+\alpha_p\right)^{k}}{\left(1+\alpha_p\right)^{k+1}-1}.
		\end{equation}
		Then $k_{b,p}$ is the minimal time under which equation \ref{weakUpperBoundEq} holds. \\
		Note that the function $\frac{\alpha_p}{p}\left(\alpha_p\abs{I_0} + \beta_p\right)
		\frac{\left(1+\alpha_p\right)^{k}}{\left(1+\alpha_p\right)^{k+1}-1}$ is monotonically decreasing in $k$ with an infimum value of
		\begin{equation}
		\label{infinityBound}
		\infl_{k \geq 0} \frac{\alpha_p}{p}\left(\alpha_p\abs{I_0} + \beta_p\right)\frac{\left(1+\alpha_p\right)^{k}}{\left(1+\alpha_p\right)^{k+1}-1} =
		\lim\limits_{k \to \infty} = \frac{\alpha_p}{p}\left(\alpha_p\abs{I_0} + \beta_p\right)\frac{\left(1+\alpha_p\right)^{k}}{\left(1+\alpha_p\right)^{k+1}-1} =
		\frac{\alpha_p}{p}\frac{\left(\alpha_p\abs{I_0} + \beta_p\right)}{1+\alpha_p}.
		\end{equation}
		Thus, when $\theta = \infty$, the minimal $b$ which satisfies equation \ref{weakUpperBoundEq} is given by equation \ref{infinityBound}, which in turn implies
		\begin{equation*}
		b_\infty = \frac{\alpha_p}{p}\frac{\left(\alpha_p\abs{I_0} + \beta_p\right)}{1+\alpha_p}.
		\end{equation*}
\\
This completes the proof of the weak upper bounds and loss upper bounds of the theorem. \\
We continue to the proof of the lower bounds. \\
		Let $\pi^*$ be an optimal policy.
		By Lemma \ref{GRProperties} (3),
		\begin{equation*}
		\mgr{z,\UC{\pi^*,s_0}{k}}{p} \geq \subplus{\mgr{z,\UC{I_0}{k}}{p} - \MSR{z,\UC{\pi^*,s_0}{k}}{p}}.
		\end{equation*}
		By Lemma \ref{GRProperties} (4) we also know that for $s_0 = (I_0, \emptyset)$,
		\begin{equation*}
		\MSR{z,\UC{\pi^*,s_0}{k}}{p} \leq \tilde{p} b(k+1).
		\end{equation*}
		
		Then, using the above together with equation \ref{GRCondition2} we have
		\begin{equation*}
		\mgr{z,\UC{\pi^*,s_0}{k}}{p} \geq \subplus{\alpha_p z + \beta_p - \tilde{p} b(k+1)} \quad, {\forall (z,k) \in \left[ \abs{I_0} ,\theta \right] \times \left\{ 0, \hdots, k_{b,\tilde{p}} \right\} }.
		\end{equation*}
		The function $f(z) = \subplus{\alpha_p z + \beta_p - \tilde{p} b(k+1)}$ is convex and monotonic non-decreasing in $z$ for all $z \geq 0$. Thus, by Remark \ref{mgrRemark} of Theorem \ref{MGRThm}, the process $M_t^{\pi^*} = \E^{\pi^*} \giventhat{\abs{I_t}}{s_0}$ is lower bounded by the solution to the recursion relation
		\begin{align*}
		&Y_{k+1} = Y_{k} + \subplus{ \alpha_p Y_{k} + \beta_p - \tilde{p} b(k+1) }  &,k \geq 0, ~Y_k \leq \theta \\
		&Y_0 = \abs{I_0}
		\end{align*}
		The rest of the proof follows similarly to the upper bound proof.

\subsection*{Proof of Corollary \ref{containmentCorollary}}
By Lemma \ref{GRProperties} (4)

\begin{equation*}
\EGR{z,\UC{\pi,s}{k}}{p} = \subplus{\EGR{z,\UC{I}{k}}{p} - \ESR{z,\UC{\pi,s}{k}}{p}}.
\end{equation*}
Since
\begin{align*}
	&\EGR{z,\UC{I_0}{k}}{p} \leq \alpha z + \beta \quad, {\forall (z,k) \in \left[ \abs{I_0} ,\theta \right] \times \left\{ 0, \hdots, k_{b,p} \right\} },\\
	&\EGR{z,\UC{I_0}{k}}{p} \geq \gamma z + \delta \quad, {\forall (z,k) \in \left[ \abs{I_0} ,\theta \right] \times \left\{ 0, \hdots, k_{b,\tilde{p}} \right\} }
\end{align*}
and
\begin{equation*}
\msr{z,\UC{I_0}{k}}{p} \leq \ESR{z,\UC{I_0}{k}}{p} \leq \MSR{z,\UC{I_0}{k}}{p},
\end{equation*}
the proof follows immediately from the proof of Theorem \ref{containmentThm1}.

\subsection*{Proof of Lemma \ref{neighborLemma}}
		By induction on $n = \abs{A}$. \\
		For $n = 1$, $A = root$, then $\abs{N\left(A\right)} = d = 1\times\left(d-1\right)+1$ and the base case holds. \\
		Suppose that for $n = k$ it holds that for any bag $A_k$ of cardinality $k$
		\begin{equation*}
		\abs{N\left(A_k\right)} = k\left(d-1\right)+1
		\end{equation*}
		We remind the reader that $root \in A_k$. \\
		Let $A_{k+1}$ be a bag of cardinality $k+1$ with $root \in A_{k+1}$.
		Next, let $v$ be a node in the deepest level of $A_{k+1}$. Then by the induction step
		\begin{equation*}
		\abs{N\left(A_{k+1} - v\right)} = k\left(d-1\right)+1
		\end{equation*}
		Since $v$ is in the deepest level of $A_{k+1}$ adding it adds exactly $d-1$ neighbors. Hence
		\begin{equation*}
		\abs{N\left(A_{k+1}\right)} = \abs{N\left(A_{k+1} - v\right)} + d-1 = \left(k+1\right)\left(d-1\right)+1
		\end{equation*}
		Finally, for any connected set $A$ in a tree, $\cut{A} = \abs{N\left(A\right)}$, thereby completing the proof of the lemma.

\subsection*{Proof of Lemma \ref{gridmgr}}

		By lemma \ref{GRProperty}, $\GR{I}{p} \geq p\abs{N(I)}$ for any set $I$. \\
		In particular when $\abs{I} = C_{d,r}$, $\GR{C_{d,r}}{p} \geq pN_{C_{d,r}}$. \\
		$\mgr{z,\UC{I_0}{k}}{p}$ is a monotonic non-decreasing function in $z$. Therefore, using the lower bound we found for $\mgr{z,\UC{I_0}{k}}{p}$ at $z = C_{d,r}$, we can lower bound $\mgr{z,\UC{I_0}{k}}{p}$ for any $z \in \left\{1, 2, \hdots \right\}$ with the following (non convex) function
		\begin{equation*}
		f(z) = pC_{d,r_z} \text{  ,where } r_z = \min \set{r}{ z \leq C_{d,r}}
		\end{equation*}
		Finally, since $f$ is not convex, we again bound it from below by a linear function in $\left\{ a, \hdots, \theta \right\}$.

\end{document}